\documentclass[%
preprint,
 superscriptaddress,
nofootinbib,
 amsmath,amssymb,
 aps,
]{revtex4-2}
\usepackage{lineno}
\newcommand*\patchAmsMathEnvironmentForLineno[1]{
  \expandafter\let\csname old#1\expandafter\endcsname\csname #1\endcsname
  \expandafter\let\csname oldend#1\expandafter\endcsname\csname end#1\endcsname
  \renewenvironment{#1}
     {\linenomath\csname old#1\endcsname}
     {\csname oldend#1\endcsname\endlinenomath}}
\newcommand*\patchBothAmsMathEnvironmentsForLineno[1]{
  \patchAmsMathEnvironmentForLineno{#1}
  \patchAmsMathEnvironmentForLineno{#1*}}
\AtBeginDocument{
\patchBothAmsMathEnvironmentsForLineno{equation}
\patchBothAmsMathEnvironmentsForLineno{align}
\patchBothAmsMathEnvironmentsForLineno{flalign}
\patchBothAmsMathEnvironmentsForLineno{alignat}
\patchBothAmsMathEnvironmentsForLineno{gather}
\patchBothAmsMathEnvironmentsForLineno{multline}
}
\usepackage{lipsum}
\usepackage{amsmath,amssymb,mathtools}
\usepackage{ascmac}
\usepackage{graphicx}
\usepackage{dcolumn}
\usepackage{bm}
\usepackage{physics}
\usepackage{enumerate}
\newcommand{\sgn}{\mathrm{sgn}\ }
\usepackage{amsthm}
\usepackage{comment}
 \newtheorem{thm}{Theorem}
 \newtheorem{defi}{Definition}
 \newtheorem{prop}{Proposition}
 \newtheorem{lem}{Lemma}
 \newtheorem{cor}{Corollary}
 \usepackage{color}
 \usepackage{ascmac}
 
  \newcommand{\tcr}{\textcolor{black}}

\newcommand\frakC{{\mathfrak C}}

\newcommand{\x}{\bm x}
\newcommand{\xeq}{\bm x^{\mathrm{eq}}}
\newcommand{\xp}{\bm x^{*}}

\begin{document}

\preprint{APS/123-QED}

\title{Local stabilizability implies global controllability\\ in catalytically-controlled reaction systems}

\author{Yusuke Himeoka}
\email[]{yhimeoka@g.ecc.u-tokyo.ac.jp}
 \affiliation{Universal Biology Institute, University of Tokyo, 7-3-1 Hongo, Bunkyo-ku, Tokyo, 113-0033, Japan}
\affiliation{Theoretical Sciences Visiting Program (TSVP), Okinawa Institute of Science and Technology Graduate University, Onna, 904-0495, Japan}

 \author{Shuhei A. Horiguchi}%
\affiliation{
    Nano Life Science Institute, Kanazawa University, Kakumamachi, Kanazawa, 920-1192, Japan
}
\affiliation{%
Institute of Industrial Science, The University of Tokyo, 4-6-1, Komaba, Meguro-ku, Tokyo
153-8505, Japan
}%

\affiliation{Theoretical Sciences Visiting Program (TSVP), Okinawa Institute of Science and Technology Graduate University, Onna, 904-0495, Japan}

\author{Naoto Shiraishi}
\affiliation{Faculty of arts and sciences, University of Tokyo, 3-8-1 Komaba, Meguro-ku, Tokyo, Japan}

\author{Fangzhou Xiao}
\affiliation{Westlake University, School of Engineering, 600 Dunyu Road, Xihu District, Hangzhou, China}


\author{Tetsuya J. Kobayashi}
\affiliation{Universal Biology Institute, University of Tokyo, 7-3-1 Hongo, Bunkyo-ku, Tokyo, 113-0033, Japan}

\affiliation{%
 Institute of Industrial Science, The University of Tokyo, 4-6-1, Komaba, Meguro-ku, Tokyo
153-8505, Japan
}%

\affiliation{%
 Department of Mathematical Informatics, Graduate School of Information Science and
Technology, The University of Tokyo, 7-3-1, Hongo, Bunkyo-ku, Tokyo 113-8656, Japan}
\affiliation{Theoretical Sciences Visiting Program (TSVP), Okinawa Institute of Science and Technology Graduate University, Onna, 904-0495, Japan}

\date{\today}
\begin{abstract}
    Controlling complex reaction networks is a fundamental challenge in the fields of physics, chemistry, biology, and systems engineering. Here, we prove a general principle for catalytically-controlled reaction systems with kinetics where the reaction order and the stoichiometric coefficient match: the local stabilizability of a given state implies global controllability within its stoichiometric compatibility class. In other words, if a target state can be maintained against small perturbations by a catalytic control, the system can be catalytically controlled from any initial condition to that state. This result highlights a tight link between the local and global dynamics of nonlinear chemical reaction systems, and clear relationship between the controllability and thermodynamic consistency of the reaction systems. The findings illuminate the robustness of biochemical systems and offers a way to control catalytic reaction systems in a generic framework. 
\end{abstract}

\maketitle
\section{Introduction}
Understanding the controllability of (bio)chemical reaction networks is crucial for both theoretical insights and practical applications in systems biology and chemical/metabolic engineering. Controllability, in the control-theoretic sense, refers to the ability to steer a dynamical system from any given initial state to any desired final state using suitable inputs. 

\tcr{In systems chemistry, networks of interacting molecules serve as programmable architectures whose emergent behaviors under out-of-equilibrium conditions can be harnessed to access specific functional states through precise chemical design \cite{Wong2017-or,Sheehan2021-er,van-Esch2017-jr}.} For biological studies, controllability translates to being able to drive a biochemical system, such as a metabolic or gene regulatory network, to a desired state by adjusting certain control parameters such as enzyme concentrations, amounts of the transcriptional/translational machineries, or external conditions.

Historically, systems biology has been rooted in control theory \cite{wiener1948cybernetics}. The early development of the field was initiated by finding feedback controls in the biochemical systems for the robust adaptation \cite{Barkai1997-by,Yi2000,Kitano2002}. This approach has been extended to construct artificial biochemical systems with a variety of functions \cite{Elowitz2000,Gardner2000,becskei2000engineering,Fung2005,Briat2016,huang2018quasi,Frei2022,SantosMoreno2020}, and for designing cell-cell interactions, including artificial cellular differentiation, pattern formation in multicellular systems, and synthetic ecosystems \cite{You2004,Kobayashi2004,Basu2005,Balagadde2008,Liu2011,Morsut2016,Toda2018}. In addition, the value of the controllability framework is well recognized in application fields such as metabolic engineering \cite{Kacser1973, Heinrich1974} and epidemiology for policy making \cite{Sharomi2017,Schnyder2025-rg}.



\tcr{Among various (bio)chemical systems, we focus on catalytic reaction systems as they provide the indispensable foundation for both biological metabolism and the bottom-up construction of out-of-equilibrium synthetic systems. The importance of catalysis stems from its ability to exert kinetic control over otherwise dormant reactions; in biological contexts, most metabolic reactions are kinetically inhibited and would not proceed on relevant timescales without enzymatic acceleration \cite{Wolfenden2001-pd}. This constraint rather enables living systems to exert precise control over their internal states by modulating catalyst activities. Similarly, in systems chemistry and bottom-up synthetic biology, catalysis serves as an essential ``control knob'' for driving systems away from equilibrium. Researchers utilize catalytic and autocatalytic feedback motifs to program functions—such as bistability, oscillations, and dissipative self-assembly—into chemical reaction networks \cite{Wong2017-or, Ashkenasy2017-ay}. By integrating these metabolic-like pathways into microcompartments, it becomes possible to assemble artificial cells with controllable metabolic functions \cite{Beneyton2018-ya}. Thus, as the functional dynamics of both natural and synthetic systems are fundamentally governed by the modulation of reaction rates through catalysis, catalytic reaction systems offer a unified and necessary framework for their control.}


\tcr{Establishing the theories on the control of (bio)chemical reactions is crucial for deepening our understanding of bottom-up constructions of artificial systems, cellular homeostasis, and biological adaptation.} While links between the passive responsiveness of biochemical reactions to external perturbations and network topologies have been actively studied \cite{okada2016law,hirono2021structural,Hirono2025-rj}, the global controllability of such catalytic reaction networks remains largely unexplored. Progress in the theory of controllability of cellular states is indispensable for deciphering mechanisms that enable biological systems to adjust flexibly to various environments \cite{Walker2023-fd,Laman_Trip2022-so,Hoehler2013-pi,Gray2019-kr,zahradka2006reassembly,hashimoto2016extremotolerant,Kikuchi2022-mx}. 

However, the application of the classical control theory to (bio)chemical reaction networks presents unique challenges. Main difficulties arise from the large number of chemical species, nonlinearity of the reaction rate functions, and non-negativity constraints on the control parameters. \tcr{The nonlinearity of reaction rate functions arises inherently from the multi-molecular reactions required to synthesize larger molecules from smaller components. Because chemical reactions are mass-conserving, in unimolecular reactions such as ${\rm A} \leftrightharpoons {\rm B}$, the substrate and product must have identical mass. However, building larger structures necessitates multi-molecular reactions\footnote{Here, ``multi'' implies that the number of molecules on either the substrate or product side is greater than one.} such as ${\rm A}+{\rm B}\leftrightharpoons {\rm C}$, where the mass of product C exceeds that of A or B individually. For such reactions, the forward reaction rate depends on the concentrations\footnote{This requirement is related to the \textit{consistency condition} in chemical reaction network theory \cite{Feinberg2019}, which guarantees that chemical concentrations do not become negative given initially non-negative concentrations.} of both A and B, thereby introducing nonlinearity into the system dynamics \cite{atkins2023atkins,Feinberg2019}.} 

\tcr{Classical control theory was originally developed for mechanical systems, such as vehicles and aircraft. In these systems, typical control inputs, such as accelerators and brakes, allow for both positive and negative acceleration. In contrast, control parameters in (bio)chemical systems—typically the concentrations or activities of catalysts—are inherently non-negative. This non-negativity constraint poses a unique challenge for evaluating the controllability of such systems \cite{haddad2010nonnegative}.}

In addition to nonlinearity and non-negativity, the evaluation of global controllability is often desired, in addition to the local controllability. These three points lead to \textit{the sign-constrained global nonlinear controllability problem}. Controllability problems in this class are highly nontrivial and unexplored by classical frameworks for the controllability of chemical reaction systems  \cite{Saperstone1973-zk,Farkas1998-zk,Dochain1992-gh,Drexler2016-bu,haddad2010nonnegative}.

In our previous work, we developed a numerical method for efficient computation of the controllability of catalytic reaction systems \cite{Himeoka2024-mo}. This method allows us to convert the sign-constrained global nonlinear controllability problem into a problem of finding appropriate conical combinations (non-negative linear combinations) of the vectors. However, we were unsuccessful in analytically identifying the controllability of the system. 

In this study, we show that local stabilizability and global controllability are tightly coupled in catalytic reaction systems. In particular, we show for a wide class of models that if a state is locally stabilizable by feedback control, then all states are globally controllable to the state by manipulating the activities of catalysts.

\tcr{The main text outlines the central concepts of global controllability in catalytic reaction systems, whereas the Supplementary Information provides the rigorous mathematical foundations in a definition-theorem-proof style.}

\section{Controllability of the catalytic reaction systems}In the present paper, we focus on a well-mixed, deterministic reaction rate equation model with $N$ chemical species and $R$ reactions. Additionally, we deal with the case in which all reactions are independently controllable to evaluate the maximum possible controllability of the system. Then, the model equation is described by the following ordinary differential equation with input-affine, no-drift control: 
\begin{eqnarray}
\dv{\x}{t}&=&\mathbb S\bm u(t)\odot [\bm v^f(\x(t))-\bm v^b(\x(t))]\nonumber \\
&=&\mathbb S\bm u(t)\odot \bm v(\x(t)),\label{eq:ode}
\end{eqnarray}
where $\bm x \in \mathbb{R}^N_{\geq 0}$ is the vector of chemical species' concentrations, and $\mathbb S$ is the $N\times R$ stoichiometric matrix. Each column of the stoichiometric matrix is called \textit{stoichiometric vectors}. The $r$th stoichiometric vector $\bm S_r$ represents the state transition of the system in the phase space by the $r$th reaction. $\bm v^f(\x)$ and $\bm v^b(\x)$ are the forward- and backward reaction fluxes, respectively. The difference between the forward and backward reaction fluxes reads the net fluxes $\bm v(\x)$. $\bm u(t):\mathbb{R}_{\ge 0}\to\mathbb{R}^R_{\ge 0}$ is the vector of time-dependent control. $\odot$ is the Hadamard (element-wise) product of vectors. In the present study, we model the net flux (the difference between the forward and backward reaction fluxes) as a single reaction flux, and the control is the modulation of catalysts' activities. Thus, the control parameters increase and decrease both the forward and backward reaction rates simultaneously, and $u_i(t)'$s are non-negative. It is not possible to control the forward and backward reaction rates independently as catalysts only change the activation energy of the chemical reactions, but do not change the chemical equilibrium \cite{atkins2023atkins}.

The ODE system (Eq.~\eqref{eq:ode}) can be used as a model equation for a wide variety of the biochemical systems. The well-mixed catalytic reaction system in the test tube regardless of whether they are biological or purely chemical, they are modeled using the equations. The metabolic reaction system in a single cell is usually modeled using the above equation \cite{Khodayari2014-ah,Boecker2021-wx,Chassagnole2002-ss,Himeoka2022-dh,Himeoka2024-my}. In addition, metabolic models using Eq.~\eqref{eq:ode} can be extended to multicellular systems in a simple manner. For the extension, we add an additional subscript on the control $\bm u$, reaction flux $\bm v$, and the concentration of metabolites $\bm x$ for indexing different cells. The exchanges of metabolites between different cells can be modeled as controlled reaction with transporters or channels as controllers.

As a simple example, \tcr{we adopt the Sel'kov model \cite{sel1968self}. The Sel'kov model is a minimalist kinetic model for glycolytic oscillations, characterized by a Hopf bifurcation and limit-cycle dynamics driven by autocatalysis. The schematic of the reaction diagram is provided in Fig.~\ref{fig:example}(a). The model comprises three reactions: $\emptyset\leftrightharpoons A$, $A\leftrightharpoons B$, and $B\leftrightharpoons \emptyset$.\footnote{\tcr{While reactions in the original model are irreversible, we treat them as reversible here. This is because, as discussed in the section IV, models consisting solely of irreversible reactions generically exhibit trivial controllability.}} The ordinary differential equations of the concentrations of $A$ and $B$, $x_A$ and $x_B$, with control is given by}
\begin{equation} 
    \tcr{ 
    \dv{}{t}\begin{pmatrix}x_A\\x_B\end{pmatrix}=\begin{pmatrix}1&-1&0\\0&1&-1\end{pmatrix}\begin{pmatrix}u_1\\u_2\\u_3\end{pmatrix}\odot\begin{pmatrix}v^{\rm max}_1 (a-k_1 x_A)\\v^{\rm max}_2(c+x_B^\gamma)(x_A-k_2x_B)\\v^{\rm max}_3(x_B-k_3b)\end{pmatrix}.\label{eq:example}
    }
\end{equation}
\tcr{
Here, $u_i, (i=1,2,3)$ are considered to be the concentrations of catalytic enzymes of the corresponding reactions. While the Sel'kov model is a coarse-grained model, if we were to assign a real-world equivalent, enzymes of reaction $1, 2$, and $3$ correspond to glucose-6-phosphate isomerase, phosphofructokinase, and fructose-bisphosphate aldolase, respectively. $v^{\rm max}_i$ are the maximum reaction speed of the corresponding enzyme. The enzyme of the second reaction (phosphofructokinase) is an allosteric enzyme and positively regulated by ADP which is represented by the chemical B in the Sel'kov model. The reaction rate increases as $x_B$ increases from the basal level $v^{\rm max}_2c$ with nonlinearity parameter $\gamma>0$. $v^{\rm max}_1a$ and $v^{\rm max}_3k_3b$ are the supply rate of A and B, respectively. $k_i'$s are the reversibility parameters. By taking $k_i\to 0$ limit for all $i'$s, all reactions become irreversible and we restore the original Sel'kov model with control. The stoichiometric vectors are drawn as the red, blue, and green arrows in Fig.~\ref{fig:example}(b).
}

The main question is how much controllability the system (Eq.~\eqref{eq:ode}) has: What are the conditions for the source state $\bm x^{\rm src}$ and the target state $\bm x^{\rm tgt}$ to have a control input $\bm u(t)$ that realizes the transition from $\bm x^{\rm src}$ to $\bm x^{\rm tgt}$ following Eq.~\eqref{eq:ode}? To answer this question, we introduce a mathematical definition of \textit{controllability}. 
\tcr{
\begin{defi}[Controllable]\label{defi:controllable}
The system is said to be asymptotically \textit{controllable} from $\bm x^{\rm src}$ to $\bm x^{\rm tgt}$ if there exists a control input $\bm u(t):\mathbb{R}_{\ge 0}\to\mathbb{R}^R_{\ge 0}$ such that the solution of Eq.~\eqref{eq:ode} with the control $\bm u(t)$ and the initial state $\bm x(0)=\bm x^{\rm src}$ reaches $\bm x^{\rm tgt}$, i.e., $\lim_{t\to\infty}\bm x(t)=\bm x^{\rm tgt}$.     
\end{defi}
While this defintion is for the asymptotic controllability in a rigorous manner, we shall use the term ``controllable'' as a shorthand for asymptotically controllable throughout this paper.} We gather all the initial states $\bm x^{\rm src}$ that can be controlled to $\bm x$, and call it \textit{the controllable set} of $\bm x$, which is denoted by ${\mathfrak C}(\bm x)$. In the present study, we identify the controllable set of a given state $\bm x$ in catalytic reaction systems, in particular, what type of state can have maximum controllability, that is, the largest possible controllable set. In the following, we restrict our attention to the positive orthant $\mathbb R^N_{> 0}$ to avoid the boundary problem \cite{Anderson2011-mo}.

Controllability is determined by the constraints imposed on the system. Here, we focus on the thermodynamic constraint that the catalysts cannot change the chemical equilibrium. Aside from this constraint, chemical reaction network systems, in general, have the network-level constraints, which are conceptualized by \textit{the stoichiometric compatibility class} (SCC) \cite{Feinberg2019}. SCC is given for each state $\bm x$ as the set of all states that can be reached from $\bm x$ by controlling the reaction rates with the sign-free control input $\bm u(t):\mathbb{R}_{\ge 0}\to\mathbb{R}^R$ in Eq.~\eqref{eq:ode}. SCC is given by the parallel translocation of the full linear span of the stoichiometric vectors;
\begin{eqnarray}
    {\cal W}(\bm x)&\coloneqq&\{\bm x+\sum_{i=1}^R a_i\bm S_i\mid a_i\in\mathbb R\}\cap \mathbb R^N_{> 0}.
\end{eqnarray} 
If a given pair of states $\bm x^{\rm src}$ and $\bm x^{\rm tgt}$ is not in the same SCC, any control is infeasible between $\bm x^{\rm src}$ and $\bm x^{\rm tgt}$. SCC of the Sel'kov model (Eq.~\eqref{eq:example}) is the whole positive orthant $\mathbb R^2_{> 0}$, while for instance, if one restrict the model to have only the second reaction $A\leftrightharpoons B$, SCC becomes the anti-diagonal line, ${\cal W}(\bm x)=\{\bm y\in \mathbb R^2_{>0}\mid y_A+y_B=x_A+x_B\}$ (see Fig.~\ref{fig:example}(b)). 

\begin{figure}[htbp]
    \begin{center}
    \includegraphics[width = 120 mm, angle = 0]{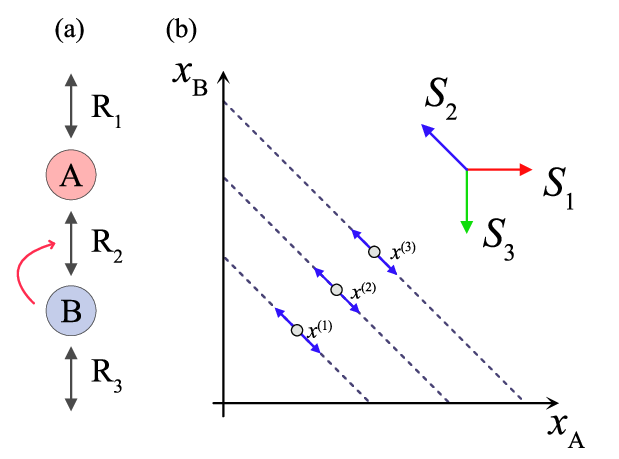}
    \caption{(a) A schematic illustration of the Sel'kov model. The black arrows represents the substrate-product relationship of each reaction, whereas the red arrow is the feedback activation. (b) The stoichiometric vectors of the reactions (red, blue, and green arrows) in the phase space. The stoichiometric compatibility class for $\bm x^{(i)}\, (i=1,2,3)$ of the model consisting only of the reaction $R_2$ are depicted as the broken lines.}
        \label{fig:example}
      \end{center}
    \end{figure}

\tcr{Lastly, we introduce the concept of \textit{local stabilizability} for the catalytically controlled chemical reaction systems for the later purpose \cite{nijmeijer1990nonlinear}.
\begin{defi}[Local stabilizability]\label{defi:stabilizable}
    A state $\bm x^*$ is said to be locally stabilizable if there exists a state-dependent feedback control $\bm u(\bm x): \mathbb{R}^N_{> 0} \to \mathbb{R}^R_{\ge 0}$ such that $\bm x^*$ is a locally asymptotically stable fixed point of the system Eq.~\eqref{eq:ode} with $\bm u(\bm x)$ as a control within its stoichiometric compatibility class $\mathcal{W}(\bm x^*)$. 
\end{defi}
Note that when we utilize a state-dependent control being only implicitly dependent on $t$ via state $\bm x$, $\bm u(\bm x(t))$, as $\bm u(t)$, Eq.~\eqref{eq:ode} becomes an autonomous dynamical system. So, the stabilizability can be rephrased with the dynamical-systems language; a system is locally stabilizable if there exists a choice of state-dependent control $\bm u(\bm x)$ such that every eigenvalues of the Jacobian matrix $$\mathbb J=\mathbb S\frac{\partial (\bm u(\bm x)\odot \bm v(\bm x))}{\partial \bm x}$$ has the negative real part except the zero eigenvalues corresponding to the left null-space of $\mathbb S$.}

\section{Catalytic control and conical combinations}In addition to the network-level constraint conceptualized by SCC, we have another constraint on the control of catalytic reaction systems. This constraint originates from the feature that the catalysts only change the activation energy of the chemical reactions \cite{atkins2023atkins}. As a consequence, the catalysts cannot directly control the directionality of the reactions, but the directions are set by the concentrations of the metabolites as the sign of the reaction rate function $v_r(\bm x)$ in Eq.~\eqref{eq:ode}. Therefore, the linear combination of stoichiometric vectors is no longer a way to evaluate the controllability of the system. However, this approach can be extended to evaluate the controllability of the catalytic reaction system. The key point is again that catalytic control does not change the chemical equilibrium, and thus, the directionality of the reaction is set by the concentrations of the metabolites. This allows us to partition the phase space into subregions where the directionality of the reaction is fixed and the conical combination of the stoichiometric vectors tells us the controllability.

The purpose of the following paragraphs is to partition phase space into regions that we call \textit{the cell} where the directionalities of the reactions are fixed. For this purpose, we assume that each reaction rate function $v_r(\bm x)$ has the following decomposition:
\begin{eqnarray}
    v_r(\bm x)&=&f_r(\bm x)p_r(\bm x)\label{eq:decomposition1}\\
    f_r(\bm x) &>& 0\label{eq:decomposition2}\\
    p_r(\bm x)&=&\prod_{i=1}^N x_i^{n^+_{i,r}}-k_r\prod_{i=1}^N x_i^{n^-_{i,r}},\label{eq:decomposition3}
\end{eqnarray}
where $n_{i,r}^\pm$ represents the reaction order of the $i$th metabolite of the forward ($n^+_{i,r}$), and the backward ($n^-_{i,r}$) reaction of the $r$th reaction. $k_r> 0$ is the reversibility constant of the $r$th reaction.\footnote{The results in the paper can be extended to the case with $\bm k\in\mathbb R^{R}_{\geq 0}$. See SI text for the details.} 
We term $\bm p(\bm x)$ and $\bm f(\bm x)$ \textit{the thermodynamic part} and \textit{the kinetic part}, respectively. Note that this assumption holds for most of the popular biochemical reaction rate kinetics such as mass-action kinetics, (generalized) Michaelis-Menten kinetics, ordered- and random multi-molecular reaction kinetics, and ping-pong kinetics \cite{Schauer1983-qg,cornish2013fundamentals}. 

An important feature of the reaction kinetics of this form is that the direction of the reaction is set by the thermodynamic part $\bm p(\bm x)$ as it is related to the thermodynamic force of the reaction. On the other hand, the remaining part $\bm f(\bm x)$ is purely kinetic, and it modulates only the absolute value of the reaction rate, but not the direction. Let us consider the Michaelis-Menten kinetics, $v_{\rm max}(x_S-kx_P)/(K_M+x_S+x_P)$ with $v_{\rm max}$, $K_M$, $k$ as the maximum reaction rate, Michaelis-Menten constant, and reversibility constant, respectively. In this case, the thermodynamic part corresponds to its numerator, $x_S-kx_P$, whereas the remaining part $v_{\rm max}/(K_M+x_S+x_P)$ is the kinetic part.

Here, we introduce the main concepts: {\it the balance manifold} ${\cal M}_r$ and {\it the cell} $C(\bm \sigma)$.\footnote{The balance manifold and cell are termed ``null-reaction manifold'' and ``direction subset'' in \cite{Himeoka2024-mo}.} 
\tcr{
\begin{defi}[Balance Manifold]\label{defi:balance_manifold}
The balance manifold of the $r$th reaction, ${\cal M}_r$ is defined as zero locus of $p_r(\bm x)$.
\begin{equation}
{\cal M}_r=\{\bm x\in \mathbb R_{> 0}^N\ |\ p_r(\bm x)=0\},\label{eq:balance_manifold}
\end{equation}
\end{defi}
}
The phase space is then partitioned by the balance manifolds $\{{\cal M}_r\}_{r=1}^R$ into regions that we call cells. In a cell $C(\bm \sigma)$, the reaction directions are fixed and represented by a binary\footnote{For the readability, we work with the setup $\bm \sigma\in \{-1,1\}^R$ of the problem in the main text, while the rigorous setup should allow $\sigma_i$ to be zero, i.e., $\bm \sigma\in \{-1,0,1\}^R$. The proof in the SI text is provided with this setup.} vector $\bm \sigma\in \{-1,1\}^R$. The cell is formally given by the following
\tcr{
\begin{defi}[Cell]\label{defi:cell}
    A subset of $\mathbb R_{>0}^N$ defined below is called a cell with reaction direction $\bm \sigma$, $C(\bm \sigma)$
    \begin{equation}
    C(\bm \sigma)=\{\bm x\in \mathbb R_{> 0}^N\ |\ {\rm sgn} \ \bm p(\bm x)=\bm\sigma\}.\label{eq:cell}
    \end{equation}
\end{defi}
}
Hereafter, ``cell'' is used to denote a region in the phase space, not a biological cell. An important note is that the geometries of ${\cal M}_r$ and $C(\bm \sigma)$ remain invariant under the catalytic control.

Because the directionality of the reactions in the cell $C(\bm \sigma)$ is $\bm \sigma$, the conical hull of the directed stoichiometric vectors $\{\sigma_i\bm S_i\}_{i=1}^R$ plays a central role in evaluating controllability within the cell. In parallel to the definition of SCC, we define \textit{the stoichiometric cone} \cite{Feinberg2019} of state $\bm x$ with the direction $\bm \sigma$ as
\begin{equation}
    {\cal V}_{\bm \sigma}(\bm x)\coloneqq\{\bm x+\sum_{i=1}^R a_i\sigma_i\bm S_i\mid a_i\geq 0\} \cap \mathbb R^N_{>0}.
\end{equation}
The set of all conical combinations of the set of vectors $\{\bm w_i\}_{i=1}^R$ is often called \textit{the conical hull} and is denoted by ${\rm cone}(\{\bm w_i\}_{i=1}^R)\coloneq \{\sum_{i=1}^R a_i\bm w_i\mid a_i\geq 0\}$. The stoichiometric cone is a parallel translation of the conical hull of the directed stoichiometric vectors. The geometry of the balance manifolds and cells of the Sel'kov model (Eq.~\eqref{eq:example}) are shown in Fig.~\ref{fig:partition}.

\begin{figure}[htbp]
    \begin{center}
    \includegraphics[width = 120 mm, angle = 0]{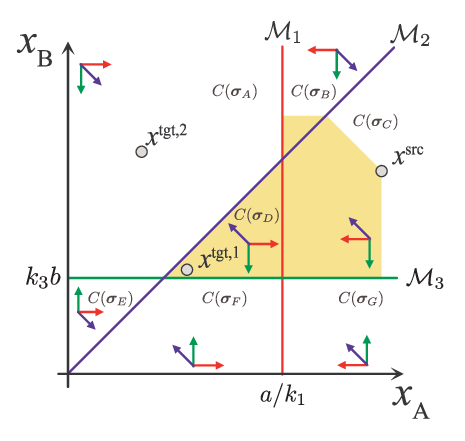}
    \caption{The phase space of the Sel'kov model (Eq.~\eqref{eq:example}) are partitioned by the balance manifolds ${\cal M}_r$ (red, blue, and green lines) into cells \tcr{$C(\bm \sigma_A), C(\bm \sigma_B),\ldots, C(\bm \sigma_G)$.} The directed stoichiometric vectors $\{\sigma_i \bm S_i\}_{i=1}^3$ are shown in each cell. The colors of the directed stoichiometric vectors are the same as the corresponding reactions. The yellow shaded region is the state reachable from $\bm x^{\rm src}$ by the non-negative controls. \tcr{One target state $\bm x^{{\rm tgt},1}$ is thus controllable from $\bm x^{\rm src}$, while the other target $\bm x^{{\rm tgt},2}$ is not. In this model, only $C(\bm \sigma_D)$ is a free cell.}}
        \label{fig:partition}
      \end{center}
    \end{figure}

Let us demonstrate how we can evaluate the controllability of the system using a cell. Assume that the trajectory of Eq.~\eqref{eq:ode} is confined within a single cell $C(\bm \sigma)$ during the interval $t\in [t_0,t_1]$. Recalling that the balance manifold and cell are independent of the control $\bm u(t)$, the trajectory is expressed as
\begin{eqnarray}
    \bm x(t) &=& \bm x(t_0)+\mathbb S \int_{t_0}^{t}\bm u(s)\odot \bm f(\bm x(s))\odot \bm p(\bm x(s))\,ds\label{eq:ray1}\\
    &=& \bm x(t_0)+\mathbb S\bm \sigma\odot \int_{t_0}^{t} \bm u(s)\odot \bm f(\bm x(s))\odot |\bm p(\bm x(s))|\,ds\label{eq:ray1a}\\
    &=& \bm x(t_0)+\mathbb S\bm \sigma\odot \int_{t_0}^{t} \tilde{\bm u}(s)\,ds\label{eq:ray2}
\end{eqnarray}
where $\tilde{\bm u}(t)={\bm u}(t)\odot\bm f(\bm x(t))\odot |\bm p(\bm x(t))|$. Since we have $f_r(\bm x)>0$ and $|p_r(\bm x)|>0$ for $\bm x \in C(\bm \sigma)$ for $1\leq r\leq R$, there is a one-to-one correspondence between $\tilde{\bm u}$ and ${\bm u}$ for fixed $\x$.

Therefore, the following statement holds: If there is a trajectory from $\bm x^{(0)}\coloneq\bm x(t_0)$ to $\bm x^{(1)}\coloneq\bm x(t_1)$ following Eq.~\eqref{eq:ode} and $\bm x(t)\in C(\bm \sigma)$ for $t\in [t_0,t_1]$, then there exists a non-negative vector $\tilde{\bm u}(t)\in \mathbb R_{\geq 0}^R$ such that 
\begin{equation}\bm x(t+\Delta t)=\bm x(t)+\mathbb S\bm \sigma\odot \tilde{\bm u}(t)\Delta t.\label{eq:dt}\end{equation}
holds in the $\Delta t\to 0$ limit for any $t\in [t_0,t_1]$. Importantly, the converse of this statement also holds. Suppose that there is a path $\bm x$ from $\bm x^{(0)}$ to $\bm x^{(1)}$ in the cell $C(\bm \sigma)$ parameterized by $t\in [t_0,t_1]$ so that $\bm x(t_0)=\bm x^{(0)}$ and $\bm x(t_1)=\bm x^{(1)}$ hold. If for any $t\in [t_0,t_1]$ there exists a non-negative vector $\tilde{\bm u}(t)\in \mathbb R_{\geq 0}^R$ that satisfies Eq.~\eqref{eq:dt} in the $\Delta t\to 0$ limit, the path $\bm x(t)$ is a solution of Eq.~\eqref{eq:ode} with the control $u_i(t)=\tilde{u}_i(t)/(p_i(\bm x(t))f_i(\bm x(t)))$. 

Based on the above arguments, the global nonlinear controllability problem is reduced to finding appropriate conical combinations in each cell. In Fig.~\ref{fig:partition}, the balance manifolds, cells, and directed stoichiometric vectors of the Sel'kov model (Eq.~\eqref{eq:example}) are presented. The controllable region from the state $\bm x^{\rm src}$ (the gray point at the right middle of the figure) is computed following the above argument, and is highlighted in yellow. If negative control is allowed, the controllable region from $\bm x^{\rm src}$ is the entire space $\mathbb R^2_{\geq 0}$, whereas the controllable region is restricted by the non-negativity of the control constraint. An intuitive restriction is that the concentration of chemical $A$ cannot increase further from the starting point $x_A^{\rm src}$ because $x_A^{\rm src}$ is larger than $a/k_1$ and the equilibrium concentration of the reaction $A\leftrightharpoons B$. By taking advantage of the mapping of the controllability problem to the conical combinations, we can develop an efficient method for numerically evaluating the controllable set. Further details on the numerical approach are provided in \cite{Himeoka2024-mo}. 

\tcr{Note that at the introduction of $\tilde{\bm u}$ (Eq.\eqref{eq:ray1a} and \eqref{eq:ray2}), the contribution of $\bm f(\bm x)$ and $|\bm p(\bm x)|$ are cancelled by defining $\tilde{\bm u}$ by $\tilde{\bm u}(t)={\bm u}(t)\odot\bm f(\bm x(t))\odot |\bm p(\bm x(t))|$. This means that the nonlinearity in the kinetic part, e.g., feedback regulation of enzyme activity, has no effect on the controllability because this activation can be counteracted by modulating $\bm u$. In the specific case of the Sel'kov model, the control cancels the feedback term $(c+x_B^\gamma)$ in Eq. \eqref{eq:example}, which plays a critical role in the Hopf bifurcation. Consequently, the Hopf bifurcation does not affect the controllability of the system.}

\section{Universal controllability of free cells} Thus far, we have shown that the controllability of catalytic reaction systems is well-captured by the conical combination inside each cell. In the following, we describe that there are special type of cells, which we call \textit{the free cells}, where the controllability is maximized: An arbitrary pair of states inside the same free cell is mutually controllable as long as the two are in the same SCC. Additionally, a given state can be locally stabilized by feedback control if and only if the state is in a free cell. Furthermore, as the main claim of the present study, an arbitrary state in the free cell is globally controllable from any initial state in the same SCC. This means that if a state is locally stabilizable, then it is globally controllable. In this section, we suppose that SCC of the model (Eq.~\eqref{eq:ode}) matches the whole space $\mathbb R^N_{> 0}$ for readability's sake, whereas we provide proof for the general case in the SI text.

\tcr{First, we formally introcude \textit{free cell};
\begin{defi}[Free Cell]\label{defi:free_cell}
A cell $C(\bm \sigma)$ is termed a free cell if the set of conical combinations of its directed stoichiometric vectors coincides with its full linear span: $${\rm cone}\{\sigma_i \bm S_i\}_{i=1}^R={\rm span}\{\sigma_i \bm S_i\}_{i=1}^R.$$ 
\end{defi}
}
In practical terms, within a free cell $C(\bm \sigma)$ the control is ``unrestricted'': For any path connecting two arbitrarily chosen states $\bm x,\bm y\in C(\bm \sigma)$ there is always a choice of $\tilde{\bm u}\in \mathbb R^R_{\geq 0}$ satisfying Eq.~\eqref{eq:dt} in $\Delta t\to 0$ limit because the conical combination of the directed stoichiometric vectors equals to their full linear span. $C(\bm \sigma_D)$ in Fig.~\ref{fig:partition} is the free cell of the Sel'kov model (Eq.~\eqref{eq:example}) and the others are non-free cells. \tcr{The reaction directions in this free cell are all positive, and the reaction flux flows from the top to the bottom of the network (Fig.~\ref{fig:example}(a)): $\emptyset\to A$, $A\to B$, and $B\to \emptyset$.} Note that the freeness of the cell is set solely by the function form of the thermodynamic part of the reaction rate function $\bm p(\bm x)$ and reversibility constant $k$. The structure of the cells is independent of the choice of, for instance, the maximum catalytic speed $v_{\rm max}$ and the Michaelis-Menten constant $K_m$ in generalized Michaelis-Menten kinetics. 

The property of the free cell offers resilience, or homeostasis in the biological context, of the system to external perturbations. Suppose that the system is perturbed within a free cell from $\bm x\in C(\bm \sigma)$ to $\bm x'\in C(\bm \sigma)$. If $C(\bm \sigma)$ is a free cell, it is always possible to control the system back to its original state $\bm x$ (see Fig.~\ref{fig:free_cell}(a)). In contrast, if $C(\bm \sigma)$ is a non-free cell, there are perturbation directions that cannot be counteracted, regardless of the perturbation strength (Fig.~\ref{fig:free_cell}(b)). This intuition leads to two useful propositions about the free cell.
\tcr{
\begin{prop}[Stabilizability in free cells]\label{prop:stabilize}
    For any state $\bm x^*\in C(\bm \sigma)$, there exists a control that locally stabilizes $\bm x^*$ within $C(\bm \sigma)$ if and only if $C(\bm\sigma)$ is a free cell.
\end{prop}
\begin{prop}[Steady-states and freeness]\label{prop:fixed_point}   
    A cell $C(\bm \sigma)$ is a free cell if and only if it admits a fixed point $\bm x^*\in C(\bm \sigma)$, either stable or unstable, under a constant control $\bm u(t)=\bm u^c\neq \bm 0$.
\end{prop}
Prop.~\ref{prop:stabilize} and \ref{prop:fixed_point} are proven in general setup in the SI text (Prop. S1 and S2).
}

\begin{figure}[htbp]
    \begin{center}
    \includegraphics[width = 140 mm, angle = 0]{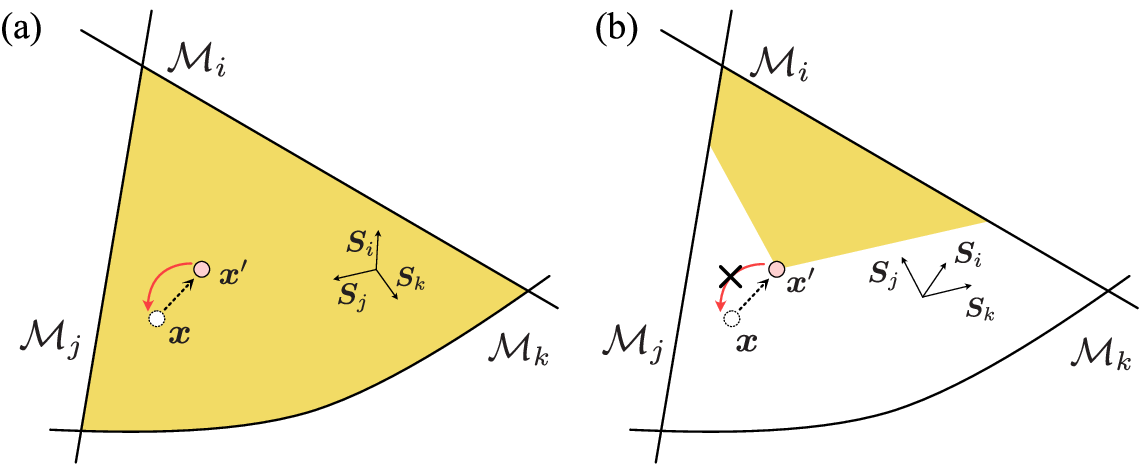}
    \caption{A visual comparison of controllability in free versus non-free cells. Yellow region is the stoichiometric cone inside the cell. (a) A perturbation from $\bm x$ to $\bm x'$ (dashed arrow) in a free cell can be counteracted, restoring $\bm x$ (solid red arc), as the conical combinations span the cell. (b) In a non-free cell, the stoichiometric cone of $\x'$ (yellow region) is insufficient to return the system to $\bm x$.}
        \label{fig:free_cell}
      \end{center}
    \end{figure}

The significance of free cells in terms of control is not only on local stabilizability, but also on global controllability. For global controllability, we need to introduce a class of reaction rate functions, which we call \textit{the stoichiometrically compatible kinetics} (SCK). SCK is a class of reaction rate functions in which the reaction order matches the corresponding stoichiometric coefficients. 
\tcr{
\begin{defi}[Stoichiometrically Compatibe Kinetics]\label{defi:SCK}
A reaction rate function $v_r(\bm x)$ is stoichiometrically compatible kinetics if it satisfies the decomposition in Eqs. \eqref{eq:decomposition1}--\eqref{eq:decomposition3} and the relation $S_{i,r}=n_{i,r}^--n_{i,r}^+$. Here;
\begin{itemize}
\item $S_{i,r}$ is the stoichiometric coefficient of the $i$-th species in the $r$-th reaction.
\item $n_{i,r}^+$ and $n_{i,r}^-$ are the reaction orders for the forward and backward reactions, respectively.
\end{itemize}
\end{defi}
}

The SCK is closely linked to the balancing condition of the chemical potential between substrates and products in the ideal-gas framework. Consider a reaction $n{\rm A}\leftrightharpoons m{\rm B}$, where $n$ and $m$ are the stoichiometric coefficients of the substrate A and product B, respectively. At chemical equilibrium, the energy balance is expressed as $n\mu_A = m\mu_B$, where $\mu_*$ denotes the chemical potential of either A or B. In the ideal-gas approximation, the chemical potential is given by $\mu_* = \mu_*^{0} + \ln x_*$, where $x_*$ is the concentration of the chemical species\footnote{Here, we set $RT = 1$, with $R$ and $T$ representing the gas constant and temperature, respectively.}. Substituting this expression for the chemical potential, taking the exponential of both sides of the equilibrium condition, and rearranging terms yields
\begin{equation}
x_A^n - k x_B^m = 0 \label{eq:energy_balance},
\end{equation}
where $k = \exp(m\mu_B^{0} - n\mu_A^{0})$. Note that the left-hand side of Eq.\eqref{eq:energy_balance} takes the form of the thermodynamic part of the reaction rate function (Eq.\eqref{eq:decomposition3}). By employing the left-hand side of Eq.\eqref{eq:energy_balance} as the reaction rate in energy-imbalanced states, we obtain a rate function implemented by SCK. While Eq.\eqref{eq:energy_balance} determines the energy-balanced state, the overall reaction rate retains a degree of freedom. This freedom is represented by the kinetic part $f(\bm x) > 0$ (Eq.\eqref{eq:decomposition2}), whose specific form depends on the details of the reaction mechanism. In the following, we refer to models in which all reactions are implemented by the SCK as \textit{full SCK models}.

Although the SCK is intimately related to chemical equilibrium, a full SCK model can attain non-equilibrium steady states. The Sel'kov model (Eq.\eqref{eq:example} and Fig.\ref{fig:example}) is a full SCK model while it can exhibit a nonzero steady reaction flux, and even more, a limit cycle oscillation. For the system to relax to a detailed-balanced chemical equilibrium, it is necessary to turn off any one of the three reactions by setting the corresponding $u_i$ to zero. The Sel'kov model relaxes to the detailed-balanced states by turning off one of the three reactions. 

The ability of the Sel'kov model to exhibit both non-equilibrium steady states and detailed-balanced chemical equilibrium by adjusting $\bm u$ is not coincidental; rather, it is a defining feature of the full SCK model, as proven in the SI text. Models that relax to detailed-balanced chemical equilibrium are in the class of \textit{thermodynamically consistent models} in the chemical thermodynamics literature \cite{Avanzini2021-ws}. The full SCK model is one from which thermodynamically consistent models can be derived by switching off the reactions \tcr{so that the number of active ($u_i\neq 0$) reactions is minimized while the rank of the stoichiometric matrix is the same with the original model (see section V in SI text).} However, the original model is not necessarily thermodynamically consistent and may exhibit non-equilibrium behaviors such as non-equilibrium steady states, oscillations, and chaos.

Our main claim is that if the reaction rate function $\bm v(\bm x)$ of all reactions in the model (Eq.~\eqref{eq:ode}) is implemented by SCK (i.e., full SCK model), then any state in arbitrary free cell can be controlled from any state (A simplified version of Corollary.~S1 in the SI text. We have Theorem.~S1 for a general setup), when formally stated,
\tcr{
\begin{thm}[Main Theorem]\label{thm:main}
    If all reactions in the model are implemented by SCK, then every state in any free cell is controllable from an arbitrary state in $\mathbb R_{>0}^N$.
\end{thm}
}

In other words, the controllable set of any state in the free cell coincides with the whole space, ${\mathfrak C}(\bm x)=\mathbb R^{N}_{> 0}$ for $\bm x\in C(\bm \sigma)$ with $C(\bm \sigma)$ being a free cell.\footnote{Recall that we supposed SCC matches to the whole space $\mathbb R^N$ for the readability. If SCC is not the whole space, the controllable set of $\bm x$ coincides to SCC, ${\mathfrak C}(\x)={\cal W}(\x)$.} 

In the following, we provide an intuitive description of the claim, while the full proof is provided in the SI text. Recall that we have $N$ chemical species and $R$ reactions, and $R$ is greater than $N$ because otherwise non-zero steady-state flux is impossible. Also, we use the term ``a vertex of a cell'' in the following meaning: The balance manifold of each reaction is the zero locus of binomial (see Eq.~\eqref{eq:decomposition3} and \eqref{eq:balance_manifold}). Thus, the balance manifolds are hyperplanes in the logarithmic-transformed positive orthant $\mathbb R_{>0}^N$, and the cells are polytope or polyhedra in the space. So, each cell has vertices at the intersection of at least $N$ balance manifolds of the reactions. 

Let us denote the target state $\bm x^{\rm tgt}\in C(\bm \sigma)$, the source state $\bm x^{\rm src}$, and $C(\bm \sigma)$ as the free cell. We take one of the vertices of the free cell $C(\bm \sigma)$ and denote it as $\xeq$. By reordering the reaction indices, we can assume that the intersection of the balance manifolds $\{{\cal M}_r\}_{r=1}^N$ is $\xeq$. Since the balance manifolds are the zero-locus of each reaction rate function, $\xeq$ is the detailed-balanced, chemical equilibrium state of the $N$ reactions. Analogous to the uniqueness and global stability of the equilibrium state in thermodynamics, it is known that the detailed-balanced, chemical equilibrium state of a given system is unique and globally asymptotically stable if the model is implemented by SCK \cite{Rao2016-tz,Schuster1989-vf,Craciun2015-zk}. Therefore, the system relaxes to $\xeq$ by the control to stop all reactions whose indices are greater than $N$.

Owing to the global stability of the chemical equilibrium, any state can reach the chosen vertex of the free cell $C(\bm \sigma)$, $\xeq$. The next step is to enter the inside of the free cell from the vertex $\xeq$. This is achieved by controlling all reactions in the system. Since $\xeq$ can be globally asymptotically stable, we can construct a locally asymptotically stable state infinitesimally close to $\xeq$ and inside the free cell, by perturbatively incorporating the contribution of the reactions that were stopped in the previous step. In this step, the feature of the free cell---the conical hull matching the full linear span---plays a crucial role. 

Once the state enters the free cell, it can be controlled to any state in the free cell, particularly to $\bm x^{\rm tgt}$ as any pair of states inside the same free cell is mutually controllable. In summary, one can control any source $\bm x^{\rm src}$ to $\bm x^{\rm tgt}\in C(\bm \sigma)$ in three steps (see Fig.~\ref{fig:procedure}):
\begin{enumerate}
    \item Control the state from $\bm x^{\rm src}$ to $\xeq$ by utilizing the global stability of the chemical equilibrium.
    \item Control the state from $\xeq$ to $\xp$ by utilizing all the reactions.  
    \item Control the state from $\xp$ to $\bm x^{\rm tgt}$ by taking advantage of the free cell property.
\end{enumerate}

All reactions in the Sel'kov model (Eq.~\eqref{eq:example}) are implemented by SCK, and thus, any state in the free cell is controllable from any state. The free cell of the model is the triangle at center filled with yellow in Fig.~\ref{fig:partition} (a cell labeled $C(\bm \sigma_D)$). Readers may see by following the arrows depicted in the figure that any state in the free cell is certainly reachable from any other state in $\mathbb R_{\geq0}^2$.

\tcr{This theorem and Prop.~\ref{prop:fixed_point} are the reason why we needed to make the Sel'kov model reversible (Eq.\eqref{eq:example}). If all reactions are irreversible, the model has only a single cell. Since the irreversible Sel'kov model has a fixed point regardless of parameter values, from Prop..~\ref{prop:fixed_point} this cell is a free cell. Therefore, an arbitrarily chosen pair of states are trivially, mutually controllable. }

\begin{figure}[htbp]
    \begin{center}
    \includegraphics[width = 80 mm, angle = 0]{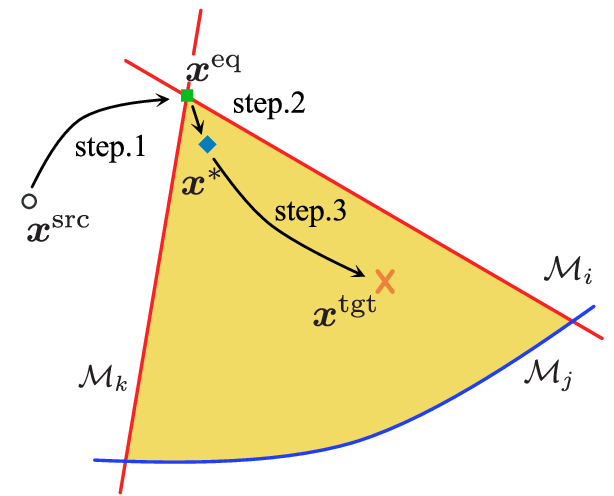}
    
    \caption{A graphical description of the control procedure. The state $\bm x^{\rm src}$ is driven to the target state $\bm x^{\rm tgt}$ via three steps. (Step 1) Control is applied to reach the equilibrium state $\xeq$ by activating the relevant control parameters (the others are set to zero). (Step 2) The control is adjusted to $\bm u^*$, controlling the state to $\xp$. (Step 3) Finally, control within the free cell brings the state to $\bm x^{\rm tgt}$. The free cell is highlighted in yellow. The states $\bm x^{\rm src}$, $\bm x^{\rm eq}$, $\bm x^*$, and $\bm x^{\rm tgt}$ are depicted by an open circle, green square, blue diamond, and orange cross, respectively. Balance manifolds intersecting at $\xeq$ (e.g., ${\cal M}_i$ and ${\cal M}_k$) are shown in red, while the other balance manifold (${\cal M}_j$) is depicted in blue.}
        \label{fig:procedure}
      \end{center}
    \end{figure}

\section{A toy model demonstration} Finally, we demonstrate global controllability using a toy model of metabolism. \tcr{The model is a reduced model of glycolysis consisting of} three metabolites ($S$, ATP, and ADP) and the following three reactions\footnote{\tcr{As stated, implementing the reaction kinetics by non-SCK is a necessary condition for the model to have a state in a free cell whose controllable set does not coincide with the whole space. However, it seemingly not easy to construct such model with a small number of chemical species. We have evaluated the controllability of 2-variable models in literature: the Sel'kov model, Schnakenberg model, and Brusselator model. Even if we make the reaction kinetics non-SCK, these three models do not exhibit uncontrollability to a state in free cells (SI text section VI).}}:
\begin{eqnarray*}
    R_1:&\quad&n\, {\rm ATP}\leftrightharpoons S+n\,{\rm ADP}\\
    R_2:&\quad&S+(n+1)\, {\rm ADP}\leftrightharpoons (n+1)\,{\rm ATP}\\
    R_3:&\quad&{\rm ATP}\leftrightharpoons {\rm ADP}.
\end{eqnarray*} 
A schematic diagram is presented in Fig.~\ref{fig:toy_model}. Here, chemical $S$ is supplied from the external environment of the system by utilizing $n$ molecules of ATP, and it is secreted to the external with converting $(n+1)$ ADPs into $(n+1)$ ATPs. Thus, through the reaction $R_1$ and $R_2$, the system obtains one ATP from one $S$. The reaction $R_3$ is the reaction for balancing ATP and ADP. \tcr{This ``investment and payoff'' architecture mirrors the design of glycolysis. Specifically, in the pathway from glucose-6-phosphate to pyruvate, one ATP molecule is initially consumed to subsequently yield four ATP molecules.} As ATP and ADP are just interconverted by the reactions, the total amount of ATP and ADP is conserved, and SCC is the level set of the total amount of ATP and ADP. Thus, we denote the total amount of ATP and ADP as $A$ and study the reduced model with two variables, given by
    \begin{eqnarray}
        &&\dv{}{t}\begin{pmatrix}
        x_S\\
        x_A
        \end{pmatrix}=
        \begin{pmatrix}
        1&-1&0\\
        -n&n+1&-1&
        \end{pmatrix}
        \begin{pmatrix}
        u_1\\
        u_2\\
        u_3
        \end{pmatrix}
        \odot
        \begin{pmatrix}
            x_A^n-k_1x_S(A-x_A)^n\\
        x_S(A-x_A)^m-k_2 x_A^m\\
        x_A-k_3(A-x_A)
        \end{pmatrix},\nonumber
    \end{eqnarray}
    where $x_S$ and $x_A$ are the concentrations of $S$ and ATP, respectively, and the kinetic part $f_i(\bm x)$ is set to unity for every $i$. The reaction rate function of the reaction $R_2$ becomes non-SCK by setting $m\neq n+1$, whereas the other reactions are implemented by SCK. We study the case in which $n=2$ in the following.

    We demonstrate that the control procedure argued above works for $m=n+1$ case, i.e., the SCK case. Our claim does not guarantee that the control procedure fails in the case of non-SCK; however, in this toy model, the procedure can indeed fail to control in the case of $m\neq n+1$. In the following, we show the result with the total concentration of ATP and ADP, $A$, equal to unity. We denote the balance manifold on SCC with $A=1$ as ${\cal M}_r$ although this is an abuse of the notation.

    First, we present a case with SCK. Here, we computed the controllability to the target state $\bm x^{\rm tgt}=(1.0,0.6)$ in the free cell $C(\bm \sigma)$ with $\bm \sigma=(1,1,1)$ (blue cell in Fig.~\ref{fig:toy_model_result}(a)) where the system steadily generates ATP through reactions $R_1$ and $R_2$. The free cell has two vertices ${\cal M}_1\cap{\cal M}_3$ and ${\cal M}_2\cap{\cal M}_3$. While we chose ${\cal M}_1\cap{\cal M}_3$ as the equilibrium state to be controlled to (see Fig.~\ref{fig:toy_model_result}(a)), the choice does not affect the result. The number and position of the vertices of the cells depend on the reaction order and reversibility parameters of the reactions. To avoid confusion, we explicitly denote the chosen state as the equilibrium state of SCK model, $\bm x^{\rm eq,SCK}$. Following the control procedure, we first set the control parameter to $\bm u^{\rm eq,SCK}=(1,0,1)$ to relax to the equilibrium state $\bm x^{\rm eq,SCK}$. The second step is to control the state inside the free cell by setting the control parameter to $\bm u^*$. Finally, we set the control parameter to the feedback control $\bm u^{\rm fb}$ which locally stabilizes $\bm x^{\rm tgt}$. As shown in Fig.~\ref{fig:toy_model_result}(b), the relaxation to $\xeq$, entrance to $C(\bm \sigma)$, and control to $\bm x^{\rm tgt}$ are successfully realized in each step.

    Next, we show that the control fails depending on the initial state in the non-SCK case. In the non-SCK case, we chose $m=1$ though $n=2$, and used the same target $\bm x^{\rm tgt}=(1.0,0.6)\in C(\bm \sigma)$ with $\bm \sigma=(1,1,1)$. In this case, the free cell has only a single vertex ${\cal M}_1\cap {\cal M}_2$ and we denote it as $\bm x^{\rm eq,nSCK}$ (see Fig.~\ref{fig:toy_model_result}(c)). As the model is non-SCK, the relaxation to $\bm x^{\rm eq,nSCK}$ with an inactivation of the reactions, except reactions $R_1$ and $R_2$ is not guaranteed. Indeed, the system relaxes toward origin $\bm 0$ instead of $\bm x^{\rm eq,nSCK}$ with $\bm u^{\rm eq,nSCK}=(1,1,0)$ as shown in the gray trajectories in Fig.~\ref{fig:toy_model_result}(d). While all trajectories relax toward the origin with $\bm u^{\rm eq,nSCK}$, some trajectories pass through the free cell $C(\bm \sigma)$ depending on the initial states. Once the state enters the free cell, it can be controlled to the target state $\bm x^{\rm tgt}$ regardless of whether the reaction kinetics are SCK or not. The cyan-colored trajectories in Fig.~\ref{fig:toy_model_result}(d) are the trajectories that enter the free cell and are controlled to $\bm x^{\rm tgt}$. 

    The control failure in the non-SCK case is not due to the control procedure adopted here. We recently developed \textit{Stoichiometric Rays} that gives an overestimate of the controllable set with a given number of reaction direction flips (the allowed count that the controlled trajectory crosses the balance manifold) \cite{Himeoka2024-mo}. Using Stoichiometric Rays, we can evaluate the controllability of the catalytic reaction system given in Eq.~\eqref{eq:ode}, from a given source state to a target state with an arbitrary, non-negative control input $\bm u(t)$. The black dots in Fig.~\ref{fig:toy_model_result}(c) are the states judged as uncontrollable to the target state $\bm x^{\rm tgt}$ with the maximum reaction flips being six by Stoichiometric Rays. 
    
    The difference in the controllability in the two cases is intuitively understood from a geometric viewpoint: In the SCK case (Fig.~\ref{fig:toy_model_result}(a)), the conical hull of the directed stoichiometric vectors \footnote{For the visibility's sake, we display the stoichiometric vectors in the logarithm-converted space. Note that the stoichiometric vectors are not straight arrows in the logarithm-transformed space.} in the cells around $\bm x^{\rm eq,SCK}$, $C(\bm \sigma), C(\bm \sigma_B), C(\bm \sigma_C)$, and $C(\bm \sigma_D)$ are directed to the $\bm x^{\rm eq,SCK}$. On the other hand, in the non-SCK case (Fig.~\ref{fig:toy_model_result}(c)), the cell with the reaction direction $\bm \sigma_D$ does not exist, and instead, $C(\bm \sigma_F)$ is a neighbor of $\bm x^{\rm eq,nSCK}$. The direction of the conical hull in $C(\bm \sigma_F)$ is not directed to $\bm x^{\rm eq,nSCK}$. This makes the approach to $\bm x^{\rm eq,nSCK}$ from $C(\bm \sigma_F)$ impossible and leads to the uncontrollablity of the states in the region represented by the black points in Fig.~\ref{fig:toy_model_result}(c) to $\bm x^{\rm eq,nSCK}$.

    \begin{figure}[htbp]
        \begin{center}
        \includegraphics[width = 80 mm, angle = 0]{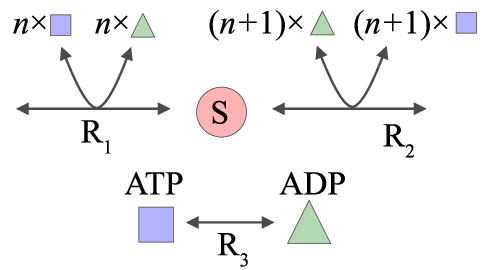}
        \caption{A schematic diagram of the toy model of metabolism. The substrate chemical S is taken up from the external environment via the reaction $R_1$ with consumption of $n$ ATPs. The reaction $R_2$ generates $(n+1)$ ATP by secreting the chemical S, respectively. The reaction $R_3$ is the reaction for balancing ATP and ADP. }
            \label{fig:toy_model}
          \end{center}
        \end{figure}

    \begin{figure}[htbp]
        \begin{center}
        \includegraphics[width = 160 mm, angle = 0]{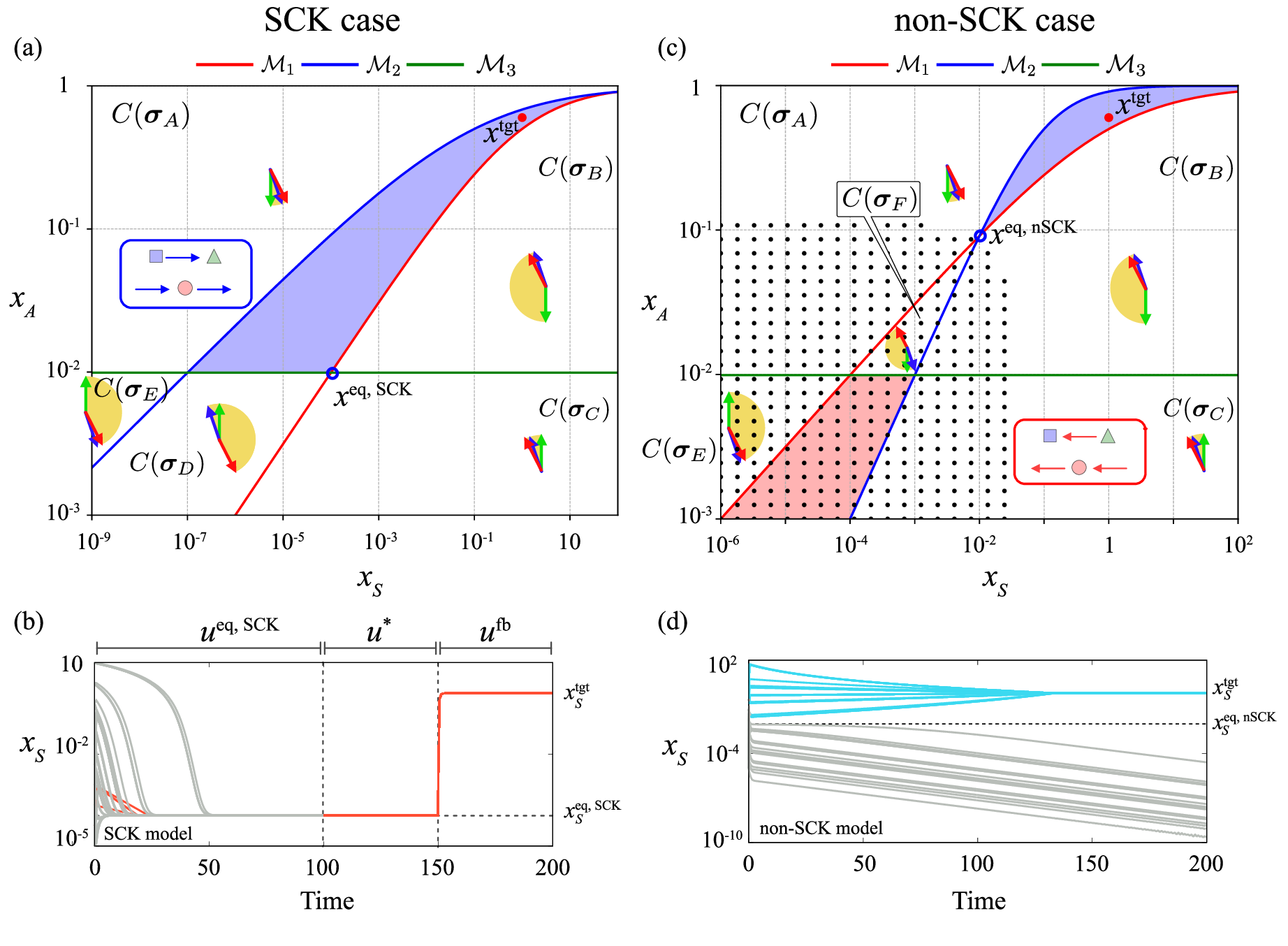}
        \caption{(a) The geometry of the cells in the phase space of SCK model. The blue shaded region is the free cells with the reaction directions shown in the inset. The directed stoichiometric vectors are depicted for the non-free cells with the red, blue, and green arrow for the reaction $R_1$, $R_2$, and $R_3$, respectively. As eye guides, angles of vectors that can be represented by a conical combination of the directed stoichiometric vectors are highlighted in yellow. (b) The controlled trajectories of the toy model with SCK ($n=2$ and $m=3$). The control input is sequentially switched from $\bm u^{\rm eq,SCK}$ to $\bm u^*$ and then to $\bm u^{\rm fb}$ at $t=100$ and $t=150$, respectively. The trajectories are colored in red if the state is in the target free cell, and otherwise are colored in gray. (c). The geometry of the cells in the phase space of non-SCK model. The cells with color shade are the free cells. The reaction directions of the red free cell is illustrated in the inset while the blue one is the same as shown in (a). The black dots are uncontrollable states to the target state regardless of the control. (d) The controlled trajectories of non-SCK ($n=2$ and $m=1$) model. The control input is initially $\bm u^{\rm eq,nSCK}$.  The trajectories are colored in cyan if the state is in the free cell $C(\bm \sigma)$. Once the state enters the free cell, the state is controlled to the target state $\bm x^{\rm tgt}$. Control parameters are set as follows; $\bm u^{\rm eq,SCK}=(1,0,1)$, $\bm u^{\rm eq,nSCK}=(1,1,0)$, and $\bm u^*\approx(4.67, 0.05, 1)$. The detailed description on the feedback control $\bm u^{\rm fb}$ is given in the SI text.}
            \label{fig:toy_model_result}
          \end{center}
        \end{figure}

    \section{Discussion} In the present study, we proved the global controllability to the states in free cells for catalytic models in which all catalysts are independently controllable and the reaction rate functions are implemented using the stoichiometrically compatible kinetics (SCK). For such a model, if the state is locally stabilizable, then the state is guaranteed to be in a free cell; thus, it is always controllable from any state in the same stoichiometric compatibility class (SCC). Taken together with the theorem in the conventional control theory that globally controllable states are locally stabilizable \cite{Clarke1997-ol}, local stabilizability and global controllability are equivalent in catalytic reaction systems. 

    By evaluating controllability, we can estimate the fundamental limits of the system. If a given \textit{in silico} model cannot attain a given task regardless of the control, it is reasonable not to expect the task to be achieved in the real system as long as the \textit{in silico} model is carefully constructed. Such overestimation of the fundamental limits of the systems may be useful for a deeper understanding of these systems. For instance, in the biological context, the fundamental limits of survival, adaptation, differentiation etc. can be estimated by the controllability of biological systems. The present results suggest that the metabolic system itself is not a limiting factor for the controllability of biological systems to homeostatically exhibit desired functions, as long as the reactions are implemented by SCK. 

    \tcr{
    The notion of chemical systems out of equilibrium emphasizes that open, energy-driven reaction networks can exhibit dynamics far from equilibrium. Systems chemistry aims to harness this principle by constructing reaction networks that operate under out-of-equilibrium conditions and use network motifs as programmable modules \cite{Wong2017-or,van-Esch2017-jr}. In peptide-based systems, dynamic features emerge from integrating multiple components: balancing strong folding against adaptivity and flexibility is critical for generating supramolecular order and disorder \cite{Sheehan2021-er}. Our finding that local stabilizability implies global controllability suggests that if a synthetic reaction network can maintain a target state against small perturbations, it can be steered to that state from any initial condition by manipulating the concentrations or activities of the catalytic molecules. This insight could contribute to the design of programmable chemical networks and help balance stability and adaptability in peptide-based systems and other out-of-equilibrium chemical reaction networks. 
    }

    However, the lesson from universal controllability---all chemical states exhibiting homeostasis are globally controllable by modulations of catalytic activity---should be taken with careful consideration of the prerequisites of control. The metabolic system is an example of catalytic reaction systems. The universal controllability, being taken at the face value, counterintuitively indicates that cellular metabolism can be controlled to any state showing homeostasis, even from an apparently ``dead'' state. Here, we discuss two aspects that we believe to be essential on the universal controllability.

    First, SCK is crucial for controllability. As mentioned, SCK is the reaction kinetics with a minimalistic extension of the detailed balance condition to non-equilibrium states with the ideal-gas formalism. While non-SCK reaction kinetics can be derived from SCK kinetics by coarse-graining multistep enzymatic reactions, the derivation needs a specific combination of the elementary steps of the reactions (see SI text Sec.~IV). However, a few mathematical models of cellular metabolism adopt non-SCK reactions for the better fits with the experimental results \cite{Boecker2021-wx,Thornburg2022-nm}. The non-SCK reaction kinetics would be better attributed to the limitation of the ideal-gas based framework rather than the specific construction of the elementary reaction steps. The limitations of the ideal-gas formalism for understanding the biological system have often been discussed, and much effort has been made for developing the chemical thermodynamics and chemical reaction network theory based upon the non ideal-gas formalism \cite{Avanzini2021-ws,Muller2014-cg,Wachtel2018-fp,Horn1972-wt,Sughiyama2022-qr,Kobayashi2024-zo}. In this sense, the deviation of the intracellular chemical reaction kinetics from the ideal-gas kinetics led by the molecular crowding, liquid-liquid phase separation, etc. would be one of the key factors of cellular controllability.\footnote{The global asymptotic stability of the detailed-balanced state is crucial for the whole proof as seen in the SI text. The universal controllability may hold even for the cases with non-SCK kinetics as long as the uniqueness and global asymptotic stability of the detailed-balanced state is guaranteed.}
    
    Another aspect is the autonomy of living systems. In the current setup for control, several unrealistic conditions were implicitly assumed, namely, perfect observability, arbitrariness of the control dynamics, and absence of the control cost. Among them, the absence of the control cost relates to disregarding the autonomy of living systems: Living systems must produce enzymes to control metabolic fluxes by themselves. The production of these enzymes requires energy and building blocks of proteins, such as ATP and amino acids, and such resources should be produced by metabolic reaction systems. 
    
    \tcr{This intrinsic requirement highlights a fundamental distinction between engineered control systems and biological organisms: while the former typically rely on external agents to set goals and manipulate parameters, the latter are characterized by their capacity for self-determination and self-maintenance \cite{Ruiz-Mirazo2004-no}. In the autonomy perspective, such self-determination is grounded in organisational closure, i.e., the mutual dependence of the system's components and operations for their production and maintenance, which collectively determine the conditions under which the system can exist \cite{Moreno2015-tr}. Regulation then appears as ``control from within'': it is exerted by dedicated subsystems that are produced by the organism and materially supported by its ongoing self-maintenance, yet are dynamically decoupled from the processes they modulate, enabling them to selectively adjust internal dynamics in response to specific perturbations \cite{Bich2016-es}. Accordingly, our control-theoretic results should be read as a first step that characterises controllability properties under idealised external actuation; extending them toward genuine biological autonomy will require endogenising the controller and its material/energetic constraints, and relating control objectives to viability norms generated by the organisation itself.
}

    Some of the authors recently proposed a mathematical framework of ``death'' \cite{Himeoka2024-mo}. In the framework, the dead state is defined as the state that is not reachable to the \textit{representative living states} which are the reference states of ``living''. Na\"{i}vely, living states have homeostasis, and otherwise we cannot observe ``living'' states because organisms are under the continuous exposure of the fluctuations. Therefore, the representative living states should have local stabilizability by control, and thus, they should be chosen from free cells. According to the present results, all states can be controlled to the representative living states indicating that death is impossible in metabolic models with the current controllability setup.
    
    For the mathematical understanding of cell deaths, the above two points would be essential, namely, the non ideal-gas based framework of metabolism and theories focusing on the autonomous nature of living systems: The machineries for controlling the system are produced by the system to be controlled.

\section*{Acknowledgments}
This work was supported by JSPS KAKENHI (Grant Numbers JP22H05403 and JP25H01390 to Y.H.; JP24K00542 and JP25H01365 to T.J.K.; JP24KJ0090 to S.A.H), JST (Grant Number JPMJCR25Q2 to T.J.K.), and Joint Research of the Exploratory Research Center on Life and Living Systems (ExCELLS) (ExCELLS program No 25EXC603-2 to YH). This research was partially conducted while visiting the Okinawa Institute of Science and Technology (OIST) through the Theoretical Sciences Visiting Program (TSVP).

\bibliographystyle{unsrt}
\bibliography{ref}

\clearpage
\setcounter{equation}{0}\renewcommand{\theequation}{S\arabic{equation}}
\setcounter{figure}{0}\renewcommand{\thefigure}{S\arabic{figure}}
\setcounter{table}{0}\renewcommand{\thetable}{S\arabic{table}}
\setcounter{thm}{0}\renewcommand{\thethm}{S\arabic{thm}}
\setcounter{defi}{0}\renewcommand{\thedefi}{S\arabic{defi}}
\setcounter{prop}{0}\renewcommand{\theprop}{S\arabic{prop}}
\setcounter{lem}{0}\renewcommand{\thelem}{S\arabic{lem}}
\setcounter{cor}{0}\renewcommand{\thecor}{S\arabic{cor}}

\setcounter{section}{0}

\section*{Supplementary Information}

\preprint{APS/123-QED}

\title{Supplementary Information Text for \textit{Local stabilizability implies global controllability in catalytically-controlled reaction systems}}

\author{Yusuke Himeoka}
\email[]{yhimeoka@g.ecc.u-tokyo.ac.jp}
 \affiliation{Universal Biology Institute, University of Tokyo, 7-3-1 Hongo, Bunkyo-ku, Tokyo, 113-0033, Japan}
\affiliation{Theoretical Sciences Visiting Program (TSVP), Okinawa Institute of Science and Technology Graduate University, Onna, 904-0495, Japan}

 \author{Shuhei A. Horiguchi}%
\affiliation{
    Nano Life Science Institute, Kanazawa University, Kakumamachi, Kanazawa, 920-1192, Japan
}
\affiliation{%
Institute of Industrial Science, The University of Tokyo, 4-6-1, Komaba, Meguro-ku, Tokyo
153-8505, Japan
}%

\affiliation{Theoretical Sciences Visiting Program (TSVP), Okinawa Institute of Science and Technology Graduate University, Onna, 904-0495, Japan}

\author{Naoto Shiraishi}
\affiliation{Faculty of arts and sciences, University of Tokyo, 3-8-1 Komaba, Meguro-ku, Tokyo, Japan}

\author{Fangzhou Xiao}
\affiliation{Westlake University, School of Engineering, 600 Dunyu Road, Xihu District, Hangzhou, China}


\author{Tetsuya J. Kobayashi}
\affiliation{Universal Biology Institute, University of Tokyo, 7-3-1 Hongo, Bunkyo-ku, Tokyo, 113-0033, Japan}

\affiliation{%
 Institute of Industrial Science, The University of Tokyo, 4-6-1, Komaba, Meguro-ku, Tokyo
153-8505, Japan
}%

\affiliation{%
 Department of Mathematical Informatics, Graduate School of Information Science and
Technology, The University of Tokyo, 7-3-1, Hongo, Bunkyo-ku, Tokyo 113-8656, Japan}
\affiliation{Theoretical Sciences Visiting Program (TSVP), Okinawa Institute of Science and Technology Graduate University, Onna, 904-0495, Japan}

\maketitle

\section{Preparation}

In the SI text, we consider the system given by

\begin{equation}
    \dv{\bm x}{t}=\mathbb S \bm u(t)\odot \bm f(\bm x(t))\odot \bm p(\bm x(t)),\label{eq:SI_ode_first}
\end{equation}
where $\bm u$ is the control input and $\bm f$ and $\bm p$ are the kinetic- and thermodynamic parts of the reaction rate function, respectively (see main text). In the following, if a given vector $\bm a$ has all positive (non-negative) elements, we denote $\bm a\succ \bm 0, (\bm a\succeq \bm 0)$. We allocate $I$ to the index set of reactions.  

We often work on the following system
\begin{equation}
\dv{\bm x}{t}=\mathbb S \bm u(t)\odot \bm p(\bm x(t)), \label{eq:SI_ode_wo_f}
\end{equation}
where the contribution of the kinetic part $\bm f(\bm x)$ is canceled by the control parameter $\bm u(t)$. Note that we assumed $\bm f(\bm x)\succ\bm 0$ for $\bm x\in \mathbb R^N_{> 0}$; therefore the cancellation does not violate the non-negativity constraint of the control parameter. 

For the readers' convenience, we recapitulate the definition of several concepts. The stoichiometric compatibility class is the parallel translation of $\Im\mathbb S$,
\begin{eqnarray}
    {\cal W}(\bm x)&\coloneq&\{\bm x+\sum_{i=1}^R a_i\bm S_i\mid a_i\in\mathbb R\}\cap \mathbb R_{> 0}^N.
\end{eqnarray} 
In the present text, we assume $\dim {\cal W}(\bm x)>0$ since no control is allowed in $\dim {\cal W}(\bm x)=0$ case. 

The balance manifold is defined as the zero locus of $p_r(\bm x)$, and the phase space is partitioned by the balance manifolds $\{{\cal M}_r\}_{r\in I}$ into cells $C(\bm \sigma)$ where the reaction directions are fixed and represented by $\bm \sigma$.
\begin{eqnarray}
    {\cal M}_r&\coloneq&\{\bm x\in \mathbb R_{> 0}^N\ |\ p_r(\bm x)=0\},\\
    C(\bm \sigma)&\coloneq&\{\bm x\in \mathbb R_{> 0}^N\ |\ {\rm sgn} \ \bm p(\bm x)=\bm\sigma\}.\label{eq:cell}
\end{eqnarray}

The controllable set ${\mathfrak C}(\bm x)$ is a set of states controllable to $\bm x$: If $\bm x_0$ is in ${\mathfrak C}(\bm x)$, there is a solution of Eq.~\eqref{eq:SI_ode_first} with the initial condition $\bm x_0$ reaching $\bm x$. 

In the SI text, we use \textit{steady-state} generically to refer to states satisfying $d\bm x/dt=\bm 0$, while the \textit{equilibrium state} is for the steady-state with vanishing reaction flux for the reactions with non-zero control, that is, the steady-state with $v_r(\x)=0$ for $u_r>0$. Note that the terms do not contain stability information. The stability of the states will be explicitly stated if necessary.

In the present text, we provide proofs for the following claims.
\begin{prop}[Stabilizability of states in free cells]\label{prop:stabilize}
    For any state $\bm x^*\in C(\bm \sigma)$, there exists a control that locally asymptotically stabilizes $\bm x^*$ within $C(\bm \sigma)\cap {\cal W}(\xp)$ if and only if $C(\bm\sigma)$ is a free cell.
\end{prop}

\begin{prop}[Steady-states and freeness]\label{prop:fixed_point}
    \tcr{
    A cell $C(\bm \sigma)$ is a free cell if and only if it admits a steady-state $\bm x^*\in C(\bm \sigma)$ under a constant control $\bm u$ supported on the active reactions (i.e., ${}^\exists \bm u \succeq \bm 0$ such that $\mathbb S (\bm u \odot \bm v(\bm x^*)) = \bm 0$ with $u_i > 0$ for all $i$ where $\sigma_i \neq 0$).}
\end{prop}
Here, we exclude steady-states made by setting all the control inputs to zero $\bm u=\bm 0$. Such marginally stable steady-states can exist in any cell, but they do not admit non-zero flux.

The formal definition of the stoichiometrically compatible kinetics (SCK) is as follows.
\begin{defi}[Stoichiometrically compatible kinetics]
    The reaction rate function $v_r(\bm x)$ is stoichiometrically compatible kinetics if it can be decomposed as 
    \begin{eqnarray}
        v_r(\bm x)&=&f_r(\bm x)p_r(\bm x),\\
        f_r(\bm x)&>&0,\\
        p_r(\bm x)&=&\prod_{i=1}^N x_i^{n^+_{i,r}}-k_r\prod_{i=1}^N x_i^{n^-_{i,r}},\ (k\geq 0)
    \end{eqnarray}
    and in addition, the $n_{i,r}^\pm$ satisfy
    $$S_{i,r}=n_{i,r}^--n_{i,r}^+$$
    where $S_{i,r}$ is the stoichiometric coefficient of the $i$th chemical in the $r$th reaction and $k_r>0$ is the reversibility parameter of the $r$th reaction.
\end{defi}

For the statement of the main theorem, boundaries of cells play a crucial role. We will show in Sec.~\ref{sec:proof_main} that for any cell there is a set of balance manifolds that are the boundaries of the cell. Having established the existence of the boundary, we define the \textit{full SCK subset} as follows:

\begin{defi}[Full SCK subset]\label{defi:sck_boundary}
    Let $J'$ be the index set of the balance manifolds $\{{\cal M}_j\}_{j\in J'}$ which are the boundaries of the cell $C(\bm \sigma)$. If there is an index set $J\subseteq J'$ such that $|J|=\rank \mathbb S$ and all reactions in $J$ are implemented by SCK, the cell $C(\bm \sigma)$ is said to have a full SCK subset within its boundaries. 
    \end{defi}

Here, we state the main theorem in the case with $\bm \sigma\in\{-1,1\}^R$ because the generic claim for $\bm \sigma\in\{-1,0,1\}^R$ is technical and may obscure the main argument. The extension of the main theorem to the case with $\bm \sigma\in\{-1,0,1\}^R$ is straightforward, and the statement and proof are provided after the proof of the main theorem.

\begin{thm}[Main theorem]\label{thm:main}
    If a free cell $C(\bm \sigma)$ has a full SCK subset within its boundaries, then the controllable set of any state in the free cell $\bm x\in C(\bm \sigma)$ coincides with the entire stoichiometric compatibility class, $\frakC(\bm x)={\cal W}(\bm x)$.
\end{thm}

\begin{cor}\label{col:main_whole}
If all reactions in the model are implemented by SCK, then every state in any free cell is controllable from an arbitrary state in the stoichiometric compatibility class of the corresponding state.
\end{cor}

\begin{cor}\label{cor:local_global}
    For a system with reactions implemented only by SCK, if state $\bm x$ is locally stabilizable, then the controllable set of $\bm x$ is the entire stoichiometric compatibility class, $\frakC(\bm x)={\cal W}(\bm x).$ 
\end{cor}

The Corollary.~\ref{col:main_whole} is a direct consequence of Theorem.~\ref{thm:main}, and Corollary.~\ref{cor:local_global} is derived from the Prop.~\ref{prop:stabilize} and Cor.~\ref{col:main_whole}. We provide proofs for propositions \ref{prop:stabilize}, \ref{prop:fixed_point}, and theorem \ref{thm:main}. 

Note that we have argued Cor.~\ref{col:main_whole} in the main text for readability. Cor.~\ref{col:main_whole} requires all the reactions in the model to be implemented by the stoichiometrically compatible kinetics (SCK), while the main theorem (Theorem~\ref{thm:main}) states a generic argument. For Theorem~\ref{thm:main} to hold, the reactions to be required to be implemented by SCK are only those that the balance manifolds are chosen as a set of boundaries of the free cell. SCK is not required for the other reactions. 

We introduce a useful concept and lemma to prove the claims.
\begin{defi}[Positive dependence \cite{Feinberg2019}]\label{defi:pos_dep}
    The set of vectors $\{\bm v\}_{i\in I}$ is positively dependent if there exist constants $\bm a \succ \bm 0$ such that
     \[
     \sum_{i\in I} a_i\bm v_i = \bm 0.
     \]
   \end{defi}  
   Recall that the free cell is defined as the cell in which the conical combinations of directed stoichiometric vectors span the entire cell. 
\begin{defi}[Free cells]\label{defi:free_cell}
    A cell $C(\bm \sigma)$ is termed a free cell if the set of conical combinations of its directed stoichiometric vectors equals their full linear span:
    $${\rm cone}\{\sigma_i \bm S_i\}_{i\in I}={\rm span}\{\sigma_i \bm S_i\}_{i\in I}.$$
\end{defi}

   \begin{lem}\label{lem:pos_dep}
    The following two statements are equivalent;
    \begin{itemize}
        \item The vectors $\{\bm v_i\}_{i \in I}$ are positively dependent.
        \item ${\rm cone}\{\bm v_i\}_{i\in I}={\rm span}\{\bm v_i\}_{i\in I}$.
    \end{itemize}
   \end{lem}
   \begin{proof}
   First, we show that ${\rm cone}\{\bm v_i\}_{i\in I}={\rm span}\{\bm v_i\}_{i\in I}$ holds if $\{\bm v_i\}_{i \in I}$ is positively dependent. 
   
   As ${\rm cone}\{\bm v_i\}_{i\in I}\subseteq{\rm span}\{\bm v_i\}_{i\in I}$ is trivial, we show that ${\rm cone}\{\bm v_i\}_{i\in I}\supseteq{\rm span}\{\bm v_i\}_{i\in I}$.
   
   For any vector $\bm v\in {\rm span}\{\bm v_i\}_{i\in I}$, we have $b_i\in \mathbb R, (i\in I)$ such that 
   \begin{equation}
       \bm v = \sum_{i\in I} b_i\bm v_i,\ b_i\in \mathbb R \label{eq:pos_comb1}
   \end{equation}
   holds. From the positive dependence of the vectors, we can construct the zero vector using the positive linear combination as $\bm 0=\sum_{i\in I} a_i \bm v_i,$ ($\bm a \succ \bm 0$). By adding $\gamma\bm 0$ to Eq.~\eqref{eq:pos_comb1}, we obtain the following 
   \begin{equation}
       \bm v+\gamma\bm 0 = \bm v=\sum_{i\in I} (b_i + \gamma a_i)\bm v_i \label{eq:pos_comb2}.
   \end{equation}
Because $\bm a\succ \bm 0$ holds, by setting $\gamma$ sufficiently large, we have $\bm b + \gamma \bm a\succ \bm 0$. Thus, $\bm v\in {\rm cone}\{\bm v_i\}_{i\in I}$ holds.

   
   Next, we show that $\{\bm v_i\}_{i \in I}$ is positively dependent if ${\rm cone}\{\bm v_i\}_{i\in I}={\rm span}\{\bm v_i\}_{i\in I}$ holds.
   
   Since for any $\rho\in I$, $-\bm v_\rho \in {\rm span}\{\bm v_i\}_{i\in I}={\rm cone}\{\bm v_i\}_{i\in I}$ holds, we have that
   \begin{equation}
       -\bm v_{\rho} = \sum_{i\in I} a^{(\rho)}_i\bm v_i, \ (a^{(\rho)}_i\geq 0),
   \end{equation}
   and thus, 
   \begin{equation}
       0 = \sum_{i\in I} \bar{a}^{(\rho)}_i\bm v_i, \label{eq:cone_comb}
   \end{equation}
where $\bar{a}^{(\rho)}_i=a^{(\rho)}_i+1$ for $\rho=i$ and otherwise $\bar{a}^{(\rho)}_i=a^{(\rho)}_i$. The sum of Eq.~\eqref{eq:cone_comb} over every $\rho\in I$ leads to 
   \begin{eqnarray*}
       0 &=& \sum_{\rho \in I}\sum_{i\in I} \bar{a}^{(\rho)}_i\bm v_i\\
   &=&\sum_{i\in I} \Bigl(1+\sum_{\rho \in I}a^{(\rho)}_i\Bigr)\bm v_i.
   \end{eqnarray*}
   As $1+\sum_{\rho \in I}a^{(\rho)}_i$ are positive for any $\rho\in I$, the vectors $\{\bm v_i\}_{i\in I}$ are positively dependent.
   \end{proof}

\section{The Proof of Prop.~\ref{prop:stabilize} and Prop.~\ref{prop:fixed_point}}
\setcounter{prop}{0}

\begin{prop}[Stabilizability of states in free cells]
    For any state $\bm x^*\in C(\bm \sigma)$, there exists a control that locally asymptotically stabilizes $\bm x^*$ within $C(\bm \sigma)\cap {\cal W}(\xp)$ if and only if $C(\bm\sigma)$ is a free cell.
\end{prop}

To show this, we use the following lemma.

\begin{lem}\label{lem:free_control}
If the system is in a free cell, for any control without non-negativity constraint \(\bm w(t):\mathbb R\to \mathbb R^R\), there exists a constrained control \(\bm u(t):\mathbb R\to \mathbb R_{\ge 0}^R\) that realizes the same control as $\bm w(t)$.
\end{lem}
\begin{proof}
Let us consider the system
\[
\dv{\bm x}{t}=\mathbb S\bm w(t)\,\odot \,\bm p(\bm x),
\]
where $\bm w(t):\mathbb R\to \mathbb R^R$ is a non-constrained control. Because the state is in the free cell, for any stoichiometric vector \(\bm S_i\), there exists a positive linear combination of the form
\begin{equation}
    -\sigma_\rho\bm S_\rho = \sum_{i\neq \rho} a_{i}^{(\rho)}\sigma_i\bm S_i, \quad (a_i^{(\rho)}> 0),
\end{equation}
Since the combinations are generally not unique, we choose one arbitrarily and adopt its expansion coefficients; without loss of generality, we may set \(a_\rho^{(\rho)}=0\).

Because $\bm x\in C(\bm \sigma)$, we have $|\bm p(\bm x)|\succ \bm 0$. Thus, we absorb the absolute value of $p_i$ into \(w_i\) and consider
\begin{equation}
\dot{\x} = \mathbb S\,\bm \sigma\odot\,\bm w(t) \label{eq:free_control}
\end{equation}
(For the equation that uses \(\bm u\), we similarly absorb the corresponding factor). Using the step function
\[
\Theta(w)=
\begin{cases}
1 & (w>0)\\[1mm]
0 & (w<0),
\end{cases}
\]
we can rewrite the right-hand side of Eq.~\eqref{eq:free_control} as follows:
\begin{eqnarray}
    \sum_{i=1}^R w_i\,\sigma_i\bm S_i &=& \sum_{i=1}^R \Theta(w_i) \,w_i\,\sigma_i\,\bm S_i + \sum_{i=1}^R \Theta(-w_i)\,w_i\,\sigma_i\,\bm S_i,\\[1mm]
    &=& \sum_{i=1}^R \Theta(w_i) \,w_i\,\sigma_i\,\bm S_i + \sum_{i=1}^R \Theta(-w_i)\,|w_i|\,(-\sigma_i\,\bm S_i),\\[1mm]
    &=& \sum_{i=1}^R \Theta(w_i) \,w_i\,\sigma_i\,\bm S_i + \sum_{i=1}^R \Theta(-w_i)\,|w_i| \Bigl(\sum_{j=1}^R a^{(i)}_{j}\sigma_j\,\bm S_j\Bigr),\\[1mm]
    &=& \sum_{i=1}^R \Bigl[\Theta(w_i)w_i+\sum_{j=1}^R \Theta(-w_j)|w_j|\,a^{(j)}_{i}\Bigr] \sigma_i\,\bm S_i.
\end{eqnarray}
Since the expression inside the brackets is non-negative, we can choose this to be our \(\bm u\). The mapping from \(\bm w\) to \(\bm u\) is smooth for $w_i$ except $w_i=0$. 
\end{proof}

We now prove Proposition~\ref{prop:stabilize}.

\begin{proof}
Consider the differential equation inside the free cell 
\begin{equation}
\dv{\x}{t} = \mathbb S\,\bm \sigma\odot\,\bm w(t), \label{eq:w}
\end{equation}
where the absolute value of $\bm p(\bm x(t))$ is lumped into $\bm w(t)$ as $\bm x\in C(\bm \sigma)$ and $|\bm p(\bm x)|\succ \bm 0$ holds. Since we are considering the dynamics inside the free cell, \(\bm w\) may have either positive or negative values (Lem.~\ref{lem:free_control}). Now, we wish to design dynamics that stabilize a state \(\xp \in C(\bm \sigma)\). Without a loss of generality, we set $\sigma_i=1$ for $\sigma_i\neq 0$ by utilizing the arbitrariness of the reaction direction.

To this end, we consider a function given by
\[
\Phi(\x,\xp) = \frac{1}{2}\sum_{r=1}^R \sigma_r \Bigl[\sum_{i=1}^N S_{ir}(x_i-x^*_i) - w_r^{(0)}\Bigr]^2,
\]
where \(w_r^{(0)}\) is selected such that \(\mathbb S\,\bm \sigma\,\odot\,\bm w^{(0)} = \bm 0\). As shown below, $\bm w^{(0)}$ corresponds to the steady-state reaction flux of the system. If the zero vector is chosen as $\bm w^{(0)}$, all reactions in the system are halted at $\bm x^*$, whereas the system has a non-zero steady flux at $\bm x^*$ if we have $\bm w^{(0)}\neq \bm 0$.

For this $\Phi$, we have
\begin{eqnarray}
    \frac{\partial \Phi}{\partial x_i} &=& \sum_{j,r}  \sigma_r S_{ir}S_{jr}(x_j-x^*_j) - \sum_{r}  \sigma_r S_{ir}w_r^{(0)}\\
    &=& \sum_{j,r} \sigma_r  S_{ir}S_{jr}(x_j-x^*_j).\label{eq:dPhi}
\end{eqnarray}
The rank of matrix \(\mathbb S\,\mathbb S^\top\) is reduced by the dimension of \({\rm coker}\,\mathbb S\). Therefore, while the states that minimize \(\Phi\) are arbitrary in the cokernel directions, once the SCC is specified, the intersection of the SCC with \(\nabla\Phi = \bm 0\) is unique. In particular, if we choose the SCC to which $\xp$ belongs, $\xp$ is the unique state that minimizes \(\Phi\) in SCC. 

On the other hand, we can rewrite Eq.~\eqref{eq:dPhi} as
\begin{eqnarray}
    \frac{\partial \Phi}{\partial x_i} &=& \sum_{r} \sigma_rS_{ir}\Bigl[\sum_{j}S_{jr}(x_j-x^*_j) - w_r^{(0)}\Bigr]
\end{eqnarray}
Thus, by designing the feedback control as 
\begin{equation}
w_r(t)=w_r(\bm x(t))= -\bigl(\bm S_r\cdot(\bm x-\xp) - w_r^{(0)}\bigr), \label{eq:w_feedback}
\end{equation} we obtain Eq.~\eqref{eq:w}. Note that the feedback $\bm w(t)$ given by Eq.\eqref{eq:w_feedback} can be negative-valued. However, since we are considering only the trajectories in the free cell, the control $\bm w(t)$ is always realized by a non-negative control $\bm u(t)$ (Lem.\ref{lem:free_control}). 

Hence, \(\Phi\) is a potential function satisfying
\begin{equation}
\dv{\x}{t}=-\nabla\Phi(\x,\xp). \label{eq:gradient}
\end{equation}
Thus, the dynamics are potential dynamics and relax to the potential minima $\xp$. Note that this feedback control is feasible inside $C(\bm \sigma)$, and the gradient dynamics (Eq.~\eqref{eq:gradient}) does not exclude the possibility that the trajectory in Eq.~\eqref{eq:gradient} reaches the boundary of the free cell $C(\bm \sigma)$. Thus, the basin of attraction for $\xp$ is not necessarily as large as $C(\bm \sigma)\cap {\cal W}(\xp)$. However, the basin of attraction has a non-zero volume because the distance between $\xp$ and the balance manifolds of the reaction with non-zero $\sigma_i$ is larger than zero as $\xp\in C(\bm \sigma)$. 

The impossibility of stabilizing the system at a state in a non-free cell is straightforward. For this purpose, we define the stoichiometric cone with strictly positive coefficients as ${\cal V}^+_{\bm \sigma}(\bm x)\coloneqq\{\bm x+\sum_{i\in I}a_i\sigma_i\bm S_i\mid a_i> 0\}$. Suppose that $C(\bm \sigma)$ is a non-free cell. We take an arbitrary state $\xp\in C(\bm \sigma)$, and consider the $\epsilon$-neighborhood of $\xp$ constrained on the stoichiometric compatibility class, $$N_\epsilon(\xp):=\{\x\in{\cal W}(\xp)\mid \|\bm x-\xp\|<\epsilon\}$$ so that $N_\epsilon(\xp)\subset C(\bm \sigma)$. Since the cell is not free, there exists $\bm x'\in N_\epsilon(\xp)$ such that $\xp \notin {\cal V}_{\bm \sigma}(\bm x')$. 

Concretely, for an arbitrary state $\bm x'\in {\cal V}^+_{\bm \sigma}(\xp)\cap C(\bm \sigma)$, $\xp \notin {\cal V}_{\bm \sigma}(\bm x')$ holds, i.e., if the state changes by occasional activations of the all reactions, the system is no longer able to return to the original state. To show this, we prove if any perturbation within a given cell can be counteracted by a control, the cell should be the free cell. The contraposition corresponds to the statement that we want to show.

Suppose that the system is originally at the state $\xp$ and transits to $\bm x'=\xp+\sum_{i\in I}a_i\sigma_i \bm S_i$ by a perturbation. Since we suppose that any perturbation is counteracted by a control, there is a control $\sum_{i\in I}b_i \sigma_i\bm S_i$ ($\bm b\succeq \bm 0$) which drag the system back to the original state $\xp$. Then, the following holds:
$$\bm 0=\sum_{i\in I}a_i \sigma_i\bm S_i+\sum_{i\in I}b_i \sigma_i\bm S_i=\sum_{i\in I}(a_i+b_i) \sigma_i\bm S_i.$$
Here, we consider the perturbation with $\bm a\succ \bm 0$. $\bm a\succ \bm 0$ and $\bm b \succeq \bm 0$ lead to $\bm a+\bm b\succ \bm 0$, and thus, the vectors $\{\sigma_i\bm S_i\}_{i\in I}$ are positively dependent. From Lem.~\ref{lem:pos_dep} and the definition of the free cell (Def.~\ref{defi:free_cell}), cell $C(\bm \sigma)$ becomes the free cell. 

It is noteworthy that the only way to realize $\dot{\bm x}=0$ in non-free cells is to set the control $u_r$ to zero for $r$ with $p_r(\bm x)\neq 0$ because the stoichiometric vectors are not positively dependent. Thus, the steady-state is not a dynamic equilibrium state in non-free cells, but rather a state in which the reactions are forced to stop by setting the control parameters to zero. This marginality is related to the system's lack of resilience to perturbations.

\end{proof}

Next, we prove the proposition.~\ref{prop:fixed_point}.

\begin{prop}[Steady-states and freeness]\label{prop:fixed_point}
    \tcr{
    A cell $C(\bm \sigma)$ is a free cell if and only if it admits a steady-state $\bm x^*\in C(\bm \sigma)$ under a constant control $\bm u$ supported on the active reactions (i.e., ${}^\exists \bm u \succeq \bm 0$ such that $\mathbb S (\bm u \odot \bm v(\bm x^*)) = \bm 0$ with $u_i > 0$ for all $i$ where $\sigma_i \neq 0$).}
\end{prop}
\begin{proof}
Let $J$ be the index set with non-zero $\sigma$, that is, $J\coloneq\{j\in I\mid \sigma_j\neq 0\}$. The steady-state condition is given by
$$\bm 0=\sum_{j\in J}\sigma_ju_jv_j(\xp)\bm S_j=\sum_{j\in J}\alpha_j\sigma_j\bm S_j,$$
where $\alpha_j=u_jv_j(\xp)$. Since $u_j>0$ and $v_j(\xp)>0$ hold, $\alpha_j$ is positive for $j\in J$, and thus, the vectors $\{\sigma_j\bm S_j\}_{j\in J}$ are positively dependent. From Lem.~\ref{lem:pos_dep}, the cell $C(\bm \sigma)$ is a free cell.

\tcr{Conversely, suppose that $C(\bm \sigma)$ is a free cell.
Let $J\coloneq\{j\in I\mid \sigma_j\neq 0\}$.
According to Lem.~\ref{lem:pos_dep}, the set of directed stoichiometric vectors $\{\sigma_j\bm S_j\}_{j\in J}$ is positively dependent. Therefore, there exist positive coefficients $\alpha_j > 0$ for all $j\in J$ such that
\begin{equation}
    \sum_{j\in J} \alpha_j \sigma_j \bm S_j = \bm 0.
\end{equation}
We choose an arbitrary state $\xp \in C(\bm \sigma)$. Since $\xp$ is in the cell $C(\bm \sigma)$, the reaction rate satisfies $\sgn(v_j(\xp)) = \sigma_j$, and thus $v_j(\xp) = \sigma_j |v_j(\xp)|$ with $|v_j(\xp)| > 0$ for $j \in J$.
We define the control input $u_j$ for $j \in J$ as
\begin{equation}
    u_j \coloneq \frac{\alpha_j}{|v_j(\xp)|}.
\end{equation}
Since $\alpha_j > 0$ and $|v_j(\xp)| > 0$, we have $u_j > 0$.
Substituting this into the reaction equation yields
\begin{equation}
    \sum_{j\in J} u_j v_j(\xp) \bm S_j = \sum_{j\in J} \frac{\alpha_j}{|v_j(\xp)|} (\sigma_j |v_j(\xp)|) \bm S_j = \sum_{j\in J} \alpha_j \sigma_j \bm S_j = \bm 0.
\end{equation}
For $j \notin J$, we have $v_j(\xp)=0$, so the term $u_j v_j(\xp) \bm S_j$ vanishes regardless of the value of $u_j$. Thus, $\xp$ is a steady-state (either stable or unstable) under the constructed positive control $\bm u$.}
\end{proof}


\section{Proof of Theorem.~\ref{thm:main}}\label{sec:proof_main}
In the following section, we provide detailed proofs of the main theorem. 
\setcounter{thm}{0}
\begin{thm}[Main theorem]
If a free cell $C(\bm \sigma)$ has a full SCK subset within its boundaries, then the controllable set of any state in the free cell $\bm x\in C(\bm \sigma)$ coincides with the entire stoichiometric compatibility class, $\frakC(\bm x)={\cal W}(\bm x)$.
\end{thm}

Let $C(\bm \sigma)$ be a free cell, and suppose that it has a full SCK subset within its boundaries. We denote the index set of the reactions by $I$, and the set difference of $I$ and $J$ by $K\coloneq I\backslash J$. Recall that the number of chemicals and reactions are $N$ and $R$, respectively.

The source and target states are denoted by $\bm x^{\rm src}$ and $\bm x^{\rm tgt}$, respectively. As no controlled dynamics is possible if $\bm x^{\rm src}$ and $\bm x^{\rm tgt}$ are not in the same SCC, we suppose that $\bm x^{\rm src}\in {\cal W}(\bm x^{\rm tgt})$ holds. 

The steps of the controls are divided into three steps; (step 1) control from $\bm x^{\rm src}$ to an equilibrium state $\xeq$, (step 2) control from $\xeq$ to a state inside the free cell $\xp \in C(\bm \sigma)$, and (step 3) control from $\xp$ to $\bm x^{\rm tgt}$. The feasibility of step 3 results from the unrestricted controllability of free cells (see the main text for details). Thus, we prove only step 1 and 2. Note that our controls are asymptotic control, and thus, the system does not reach exactly the target state of each control step in finite time. 

\subsection{Step 1: control from $\bm x^{\rm src}$ to $\xeq$}

First, we show that the unique equilibrium state $\xeq$ exists as the intersection of the balance manifolds $\{{\cal M}_j\}_{j\in J}$ and stoichiometric compatibility class ${\cal W}(\bm x^{\rm tgt})$. For that, it is useful to introduce the reaction order matrix $\mathbb T$, given by
\begin{equation}
\mathbb T_{i,r}=n_{i,r}^--n_{i,r}^+, \label{eq:order_matrix}
\end{equation}
where $n_{i,r}^\pm$ are the reaction orders of the forward ($+$) and backward ($-$) reactions of the chemical $i$ in reaction $r$, respectively. The $r$th row vector of the reaction order matrix $\mathbb T$ is denoted by $\bm T_r$ and is called the reaction order vector. Note that if the $r$th reaction is implemented by SCK, $\bm S_r=\bm T_r$ follows.

The definitions of the cells and balance manifolds are rewritten using the reaction order matrix $\mathbb T$ as follows:
\begin{eqnarray}
    {\cal M}_r&=&\{\bm x\in \mathbb R_{> 0}^N\ |\ \bm T^\top_{r}\ln \bm x+\ln k_r=0\},\\
    C(\bm \sigma)&=&\{\bm x\in \mathbb R_{> 0}^N\ |\ {\rm sgn}(\mathbb T^\top\ln \bm x+\ln \bm k)=\bm \sigma\}.
\end{eqnarray}

We denote the reaction order matrix and reversibility parameter vector consisting of the reactions in $J$ by $\mathbb T_J$ and $\bm k_J$, respectively. The intersection of the balance manifolds in $J$ is given by
\begin{eqnarray}
\bigcap_{j\in J}{\cal M}_j=\{\bm x\in \mathbb R_{>0}^N \mid \mathbb T^\top_J\ln \bm x=-\ln \bm k_J\}, 
\end{eqnarray}
and denoted as ${\cal M}_J$. At any state in $\cal M_J$, all the reaction fluxes of reactions in $J$ are zero, and thus, the system is in a detailed-balanced state of the system only with the reactions in $J$. For $\cal M_J$ to be non-empty, the equation 
\begin{equation}
    \mathbb T^\top_J\ln \bm x=-\ln \bm k_J \label{eq:balance_equation_T}
\end{equation}
should have a solution. A sufficient condition for Eq.~\eqref{eq:balance_equation_T} to have a solution is that the rank of the reaction order matrix $\mathbb T_J$ equals the number of reactions in $J$, that is, $\rank \mathbb T_J=|J|$. For that, we have a following useful proposition:

\begin{prop}\label{prop:facet}
Suppose \(C(\bm \sigma)\) is non-empty. There is an index set \(J \subset I,\, |J|=\rank \mathbb T\) such that:
\begin{enumerate}
  \item The matrix \(\mathbb{T}_J\) satisfies $\rank \mathbb{T}_J \;=\;\rank \mathbb{T}.$
  \item The balance manifolds \(\{\mathcal{M}_j\}_{j\in J}\) are boundary of the cell $C(\bm \sigma)$.
\end{enumerate}
\end{prop}
The proof of the proposition is provided at the end of this section. 

Following the assumption of theorem \ref{thm:main}, we assume that the free cell $C(\bm \sigma)$ has a full SCK subset within its boundaries $\{{\cal M}\}_{j\in J}$. Thus, $\mathbb T_J$ can be replaced by $\mathbb S_J$, and we have $\rank\mathbb S_J=\rank\mathbb S$. The intersection ${\cal M}_J$ is then given by 
\begin{eqnarray}
{\cal M}_J=\{\bm x\in \mathbb R_{>0}^N \mid \mathbb S_J^\top\ln \bm x=-\ln \bm k_J\}, \label{eq:balance_sub_manifold}
\end{eqnarray}
and $\rank \mathbb S_J=|J|$ holds, where $\bm k_J$ is the vector consisting of $k_j\, (j\in J)$. Because $\mathbb S_J^\top$ is column full-ranked and $\rank \mathbb S\leq N$ holds, the equation 
\begin{equation}
    \mathbb S^\top_J\ln \bm x=-\ln \bm k_J,\label{eq:wegscheider}
\end{equation}
has solutions, and thus, ${\cal M}_J$ is not empty. Eq.~\eqref{eq:wegscheider} provides $\rank \mathbb S_J=\rank \mathbb S$ independent conditions, leading to ${\rm dim}[{\cal M}_J]=N-\rank \mathbb S$. 

Let ${\cal W}_J(\bm x)$ be the SCC of an arbitrary state $\bm x\in \mathbb R_{>0}^N$ with reactions in $J$, i.e., 
\begin{equation}
{\cal W}_J(\bm x)=\{\bm x+\sum_{j\in J}a_j \bm S_j \mid a_i\in \mathbb R\}\cap \mathbb R^N_{>0}. \label{eq:SCC_J}
\end{equation}
Note that ${\cal W}_J(\bm x)$ does not necessarily have the same dimension as the original SCC, ${\cal W}(\bm x)$, because the reactions in $K$ are not included. However, owing to Prop.~\ref{prop:facet}, we have ${\rm dim}[{\cal W}_J(\bm x)]={\rm dim}\,\Im \mathbb S_J=\rank \mathbb S_J=\rank \mathbb S$, and which equals to ${\rm dim}[{\cal W}(\bm x)]$. By combining this result with the dimensionality argument of ${\cal M}_J$, we have 
\begin{equation}
{\rm dim}[{\cal W}_J(\bm x)]+{\rm dim}[{\cal M}_J]=N.\label{eq:dimension_sum}
\end{equation}

In addition, for any states $\bm x_1,\bm x_2\in {\cal W}_J(\bm x)$ and $\bm y_1,\bm y_2 \in {\cal M}_J$, the following holds:
\begin{equation}
    \langle \bm x_1-\bm x_2,\ln \bm y_1-\ln \bm y_2\rangle=\langle \mathbb S_J \bm a,\ln \bm y_1-\ln \bm y_2\rangle=\langle \bm a,\mathbb S_J^\top(\ln \bm y_1-\ln \bm y_2)\rangle=0, \label{eq:dual_product}
\end{equation}
where $\langle \cdot,\cdot \rangle$ represents the dual product between the concentration space and logarithm-transformed concentration space. In the first equality, we have used the fact that $\bm x_1$ and $\bm x_2$ are in the same SCC, and thus, there is a vector $\bm a\in \mathbb R^{|J|}$ satisfying $\mathbb S_J \bm a=\bm x_1-\bm x_2$. The second equality is due to the fact that $\bm y_1$ and $\bm y_2$ are in ${\cal M}_J$, and thus, satisfy Eq.~\eqref{eq:wegscheider}. From Eq.~\eqref{eq:dimension_sum} and Eq.~\eqref{eq:dual_product}, ${\cal W}_J(\bm x)$ and ${\cal M}_J$ are dually orthogonal complements of each other. The intersection of dually orthogonal complements with Eq~\eqref{eq:dimension_sum} always exists and is unique \cite{Kobayashi2024-zo,Sughiyama2022-qr}. We denote the intersection of ${\cal W}_J(\bm x^{\rm tgt})$ and ${\cal M}_J$ by $\xeq$.


The reactions in $J$ are exactly balanced, and thus, the thermodynamic part of the reactions in $J$ is zero at $\xeq$. Also, $\xeq$ is an adherent point of the free cell. Thus, an arbitrarily small displacement from $\xeq$ can lead to the entrance of the free cell, that is, there is a vector $\bm a$ with an infinitesimally small norm satisfying ${\rm sgn}\ \bm p(\xeq+\mathbb S \bm a)=\bm \sigma$. As this balanced state can be seen as the detailed-balanced state for the reaction system (Eq.~\eqref{eq:SI_ode_first}) with $u_j>0$ for $j\in J$ while $u_k=0$ for $k\in K$, we term it the equilibrium state. 

Let us consider the reduced model with only reactions in $J$
\begin{equation}
    \dv{\bm x}{t}=\mathbb S_J\bm u_J\odot \bm p_J(\bm x), \label{eq:reduced_model}
\end{equation}
where $\bm u_J\in \mathbb R^{|J|}_{>0}$ and $\bm p_J(\bm x)$ are the control inputs for the reaction in $J$, and thermodynamic part of the reactions in $J$, respectively. Here, we have 
\begin{equation}
    {\cal W}_J(\bm x)={\cal W}(\bm x) \label{eq:same_SCC}
\end{equation} 
for any $\bm x\in \mathbb R_{> 0}^N$ because $\Im\mathbb S_J\subseteq \Im \mathbb S$ and $\rank \mathbb S_J=\rank \mathbb S$ hold.

Since $\bm u_J$ is a constant vector, Eq.~\eqref{eq:reduced_model} can be considered as an autonomous dynamical system with the mass--action kinetics. Also, because $\xeq$ satisfies Eq.~\eqref{eq:wegscheider}, $\xeq$ is the detailed-balanced state of Eq.~\eqref{eq:reduced_model}. Thus, $\xeq$ is the globally asymptotically stable steady-state, particularly, the equilibrium state of the system within ${\cal W}_J(\bm x^{\rm tgt})={\cal W}(\bm x^{\rm tgt})$ \cite{Rao2016-tz,Schuster1989-vf,Craciun2015-zk}. The original model (Eq.~\eqref{eq:SI_ode_wo_f}) is reduced to Eq.~\eqref{eq:reduced_model} by setting $u_i(t)>0$ for $i\in J$ and $u_i(t)=0$ for $i\in K$. Let us denote one such control as $\bm u^{\rm eq}$. 

Overall, any state in ${\cal W}(\bm x^{\rm tgt})$ can be asymptotically controlled to $\xeq$ by setting $\bm u(t)=\bm u^{\rm eq}$.



\subsection{Step 2: control from a neighboring state of $\xeq$ to $\xp$}

\subsubsection{Overview of the proof of Step 2}
First, we divide Eq.~\eqref{eq:SI_ode_wo_f} into the two parts, namely, the term generated by the reaction with indices in $J$ and $K$, respectively. We denote the stoichiometric matrices, controls, and thermodynamic part of the reactions as \(\mathbb S_J\), \(\mathbb S_K\), \(\bm u_J(t)\), \(\bm u_K(t)\), \(\bm p_J(\x)\), and \(\bm p_K(\x)\), respectively. The differential equation is given by 

\begin{eqnarray}
    \dv{\x}{t}&=&\mathbb S_J \bm u_J(t)\odot \bm p_J(\x) +\mathbb S_K \bm u_K(t)\odot \bm p_K(\x) \label{eq:eq1}
\end{eqnarray}
In the following, we show the existence of control to $\xp\in C(\bm \sigma)$ from neighbors of $\xeq$ by utilizing the stability of an equilibrium state $\xeq$. For this purpose, we consider the linearized dynamics of Eq.~\eqref{eq:eq1} around $\xp$ which is sufficiently close to $\xeq$.

Let $\xp$ be a neighbor of $\xeq$ and be in the free cell $C(\bm \sigma)$. The relationship between $\xp$ and $\xeq$ is given by
\begin{equation}
\xp \coloneqq \xeq + \epsilon\bm c \quad (\|\bm c\|=1,\ \epsilon\ll 1),\label{eq:defi_xp}
\end{equation}
where $\bm c$ is selected such that \(\xp\in C(\bm \sigma)\). By expanding Eq.~\eqref{eq:eq1} around \(\xp\) and retaining the terms up to the first order in \(\Delta \x \coloneqq \x - \xp\), we obtain
\begin{eqnarray}
    \dv{\Delta \x}{t} &=& \Bigl[\mathbb S_J \bm u_J(t)\odot \bm p_J(\xp) +\mathbb S_K \bm u_K(t)\odot \bm p_K(\xp)\Bigr] \nonumber \\
    &&+\Bigl[\mathbb S_J\,\mathrm{diag}\{\bm u_J(t)\}\frac{\partial \bm p_J}{\partial \x}(\xp) + \mathbb S_K\,\mathrm{diag}\{\bm u_K(t)\}\frac{\partial \bm p_K}{\partial \x}(\xp)\Bigr]\Delta \x\label{eq:eq2}
\end{eqnarray}
Hereafter, we denote
\[
{\cal J}_X(\bm u_X,\xp) \coloneqq \mathbb S_X\,\mathrm{diag}\{\bm u_X(t)\}\frac{\partial \bm p_X}{\partial \x}(\xp), \quad (X=J,K).
\]
We say ${\cal J}_X(\bm u_X,\xp)$ is stable if the real parts of all eigenvalues of ${\cal J}_X(\bm u_X,\xp)$ are negative except for the zero eigenvalue corresponding to the left null space of $\mathbb S_X$. For the sum of the matrices, \({\cal J}_J(\bm u_J,\xp)+{\cal J}_K(\bm u_K,\xp)\), we say that it is stable if the eigenvalue of the sum of the matrix satisfies the above condition. In this case, we allow the eigenvalues corresponding to the left null space of $\mathbb S$ to be zero. 

The condition for $\Delta\bm x=\bm 0$ being stable is the vanishment of the first term on the right-hand side of Eq.~\eqref{eq:eq2} and the stability of the second term. We will show the following two claims: First, there are constants \(\bm u_J(t)=\bm u^*_J\) and \(\bm u_K(t)=\bm u^*_K\) \footnote{As an abbreviation, we also denote the concatenated vector of $\bm u_J(t)$ and $\bm u_K(t)$ by $\bm u(t)=\bm u^*$} which leads to the vanishment of the first term on the right-hand side in Eq.~\eqref{eq:eq2}. Second, the matrix in the second term \({\cal J}_J(\bm u_J,\xp)+{\cal J}_K(\bm u_K,\xp)\) is stable with such $\bm u^*_J$ and $\bm u^*_K$. The local asymptotic stability of $\Delta \bm x=\bm 0$ means that $\bm x=\xp$ is locally asymptotically stable in Eq.~\eqref{eq:eq1}. The basin of attraction to $\Delta \bm x=\bm 0$ is at least as large as the range of linearization of the dynamics (Eq.~\eqref{eq:eq2}) being validated. Thus, when the state reaches sufficiently close to $\xeq$ by the control in the step 1, we can switch the control from $\bm u^{\rm eq}$ to $\bm u^*$ to asymptotically control the state to $\xp$.

\subsubsection{Vanishment of the first term}\label{sec:vanish_app}
First, we show the vanishment of the first term. According to the definition of the free cell, we have ${\rm cone}\{\sigma_i\bm S_i\}_{i\in I}={\rm span}\{\sigma_i\bm S_i\}_{i\in I}$. From Lemma.~\ref{lem:pos_dep}, $\{\sigma_i\bm S_i\}_{i\in I}$ is positively dependent, and thus, there is a positive linear combination of $\{\sigma_i\bm S_i\}_{i\in I}$ leading to the zero vector, that is
\begin{equation}
\sum_{i\in I} w_i \sigma_i \bm S_i = \bm 0,(w_i>0).\label{eq:pos_comb_vanish}    
\end{equation}
By comparing Eq.~\eqref{eq:pos_comb_vanish} with the first term in Eq.~\eqref{eq:eq1} being zero, 
\begin{eqnarray}
    &&\mathbb S_J \bm u_J^*\odot \bm p_J(\xp) +\mathbb S_K \bm u_K^*\odot \bm p_K(\xp)=\bm 0\label{eq:vanish}\\
    &\Leftrightarrow&\sum_{i\in I}u_i^*p_i(\xp)\bm S_i=\bm 0
\end{eqnarray}
we have $u_i^*=w_i/|p_i(\xp)|>0$ for $i\in I$. 

For the later proof, we show that the leading term of $\bm u_K^*$ is of the order of $\epsilon$ if we take $\bm u_J^*$ so that its leading term is the $0$th order of $\epsilon$, i.e, $\bm u_J^*=\bm u_J^{(0)}+{\cal O}(\epsilon)$. Intuitively, this is because $\bm p_J(\xp)= {\cal O}(\epsilon)$ and $\bm p_K(\xp)= {\cal O}(1)$ hold, and thus, $\bm u_K^*$ should be ${\cal O}(\epsilon)$ if $\bm u_J^*$ is chosen to be ${\cal O}(1)$.

To verify this claim, we expand $\bm p_J(\xp)=\bm p_J(\xeq+\epsilon\bm c)$ in Eq.~\eqref{eq:vanish} on $\epsilon$. Then, we have
\begin{eqnarray}
    \bm 0&=&\mathbb S_J\,(\bm u_J^{(0)}+{\cal O}(\epsilon)) \odot \bm p_J(\xeq+\epsilon\bm c)+\mathbb S_K \bm u_K^*\odot \bm p_K(\xp)\\
    &=& \Bigl(\mathbb S_J\,(\bm u_J^{(0)}+{\cal O}(\epsilon)) \odot \bm p_J(\xeq) + \epsilon\,\mathbb S_J\,\mathrm{diag}\{\bm u_J^{(0)}\}\frac{\partial \bm p_J}{\partial \x}(\xeq)\bm c + {\cal O}(\epsilon^2)\Bigr)+\mathbb S_K \bm u_K^*\odot \bm p_K(\xp)\nonumber\\
    &=& \Bigl(\epsilon\,\mathbb S_J\,\mathrm{diag}\{\bm u_J^{(0)}\}\frac{\partial \bm p_J}{\partial \x}(\xeq)\bm c + {\cal O}(\epsilon^2)\Bigr)+\mathbb S_K \bm u_K^*\odot \bm p_K(\xp)\\
    &=&\epsilon \bm q+\mathbb S_K \bm u_K^*\odot \bm p_K(\xp) \label{eq:q}
\end{eqnarray}
where $\bm q= \mathbb S_J\,\mathrm{diag}\{\bm u_J^{(0)}\}\frac{\partial \bm p_J}{\partial \x}(\xeq)\bm c + {\cal O}(\epsilon)$ and note that $\bm p_J(\xeq)=\bm 0$ holds. 

From Eq.~\eqref{eq:q}, we have
$$-\epsilon\bm q\in{\rm cone}\{\bm S_i p_i(\xp)\}_{i\in K}={\rm cone}\{\bm S_i \sigma_i\}_{i\in K},$$
in particular, $-\bm q\in{\rm cone}\{\bm S_i \sigma_i\}_{i\in K}.$ Thus, there is a non-negative constant $a_i$ such that
\begin{equation}
    -\bm q = \sum_{i\in K} a_i \bm S_i \sigma_i.
\end{equation}
Note that $a_i$ remains ${\cal O}(1)$ as $\epsilon\to 0$ since the leading term of $\bm q$ is independent of $\epsilon$. The choice of $\bm a$ is not necessarily unique. We choose one arbitrarily and multiply both sides by $\epsilon$. By comparing this with the coefficient of the second term in Eq.~\eqref{eq:q}, we obtain
\begin{equation}
\epsilon a_i=u_i^*p_i(\xp), (i\in K),
\end{equation}
and $u_i^*$ is given by $$u_i^*=\epsilon \frac{a_i}{p_i(\xp)}, \  (i\in K).$$ Because $p_i(\xp)\to p_i(\xeq)\neq 0$ holds as $\epsilon\to 0$, we have $u_i^*={\cal O}(\epsilon)$. Henceforth, to emphasize that \(\bm u_K^*\) is of the order \(\epsilon\), we write
\[
\bm u_K^*=\epsilon\,\bm u_K^\circ.
\]

\subsubsection{Stability of the second term}\label{sec:stability_app}

Next, we show that the second term is stable. We evaluate the change in the eigenvalues of the matrix \({\cal J}_J(\bm u_J^*,\xp)+{\cal J}_K(\bm u_K^*,\xp)\) by perturbatively incorporating the effect of \({\cal J}_K(\bm u_K^*,\xp)\) into \({\cal J}_J(\bm u_J^*,\xeq)\). Since $\xeq$ is asymptotically stable, the eigenvalues of \({\cal J}_J(\bm u_J^*,\xeq)\) are all negative except for the zero eigenvalue corresponding to the left null space of \(\mathbb S_J\). The changes in the eigenvalues with the non-zero real part are evaluated by utilizing the generalized Bauer--Fike theorem \cite{Chu1986-mj}. 

Let $\|\cdot \|_p$ be the $p$-norm of the matrix and $\kappa_p(V)$ be the condition number of a matrix \(V\) induced by the $p$-norm. Suppose that $A\in \mathbb C^{N,N}$ is a matrix with eigenvalues \(\{\lambda_i\}_{i=1}^N\) and is converted to 
$$\tilde{A}=V^{-1}A V={\rm diag}(A_i),$$ where $A_i$'s are triangular Schur form. Suppose that the matrix $A$ is perturbed by matrix $B\in \mathbb C^{N,N}$ as \(A+\delta B\), $(\delta>0)$, and the eigenvalues of $A+\delta B$ are denoted by \(\{\mu_i\}_{i=1}^N\). The generalized Bauer--Fike theorem states that for any $i\in\{1,2\ldots,N\}$, there exists $j\in \{1,2\ldots,N\}$ such that

\begin{equation}
    |\mu_i-\lambda_j| \leq \max\{\theta,\theta^{1/q}\}\label{eq:gbf} 
\end{equation}
holds with
\begin{equation}
\theta=\delta C\kappa_p(V)\|B\|_p,\label{eq:gbf_theta}
\end{equation}
where $C$ and $q$ are positive, finite-valued constants computed from the Schur blocks of the matrix $A$ (see \cite{Chu1986-mj} for details). Note that if $A$ is diagonalizable, $C$ and $q$ become unity, and the classical Bauer--Fike theorem \cite{bauer1960norms} is recovered.


Here, we apply the generalized Bauer--Fike theorem to the matrix \({\cal J}_J(\bm u_J^*,\xp)+{\cal J}_K(\bm u_K^*,\xp)\). As \(\xp\coloneqq\xeq+\epsilon\,\bm c\) and \(\epsilon \ll 1\), we have
\begin{eqnarray*}
    {\cal J}_J(\bm u_J^*,\xp)
    &=& {\cal J}_J(\bm u_J^*,\xeq+\bm c\,\epsilon)\\
    &=& {\cal J}_J(\bm u_J^*,\xeq) + \epsilon \sum_{i=1}^N\frac{\partial}{\partial x_i}{\cal J}_J(\bm u_J^*,\x)\Big|_{\x=\xeq}c_i + {\cal O}(\epsilon^2).
\end{eqnarray*}
Also, for \({\cal J}_K(\bm u_K^*,\xp)\) we have
\begin{eqnarray*}
    {\cal J}_K(\bm u_K^*,\xp)
    &=& {\cal J}_K(\epsilon\,\bm u_K^\circ,\xeq+\epsilon\,\bm c)\\
    &=& \epsilon\,{\cal J}_K(\bm u_K^\circ,\xeq+\epsilon\,\bm c)\\
    &=& \epsilon\,{\cal J}_K(\bm u_K^\circ,\xeq) + {\cal O}(\epsilon^2).
\end{eqnarray*}
Retaining terms up to the first order in \(\epsilon\), we have
\begin{eqnarray}
    \tilde{\cal J}_\epsilon(\bm u_J^*,\bm u_K^*,\xeq) &\coloneqq& {\cal J}_J(\bm u_J^*,\xeq) + \epsilon \sum_{i=1}^N\frac{\partial}{\partial x_i}{\cal J}_J(\bm u_J^*,\x)\Big|_{\x=\xeq}c_i + \epsilon\,{\cal J}_K(\bm u_K^\circ,\xeq) 
    \label{eq:JIJ}
\end{eqnarray}
We denote the first order of $\epsilon$ term as $$\Delta(\bm u_J^*,\bm u_K^*,\xeq)=\sum_{i=1}^N\frac{\partial}{\partial x_i}{\cal J}_J(\bm u_J^*,\x)\Big|_{\x=\xeq}c_i + \,{\cal J}_K(\bm u_K^\circ,\xeq)$$
in the following.

Now we apply the generalized Bauer--Fike theorem. Here, we denote $\delta$-independent part in Eq.~\eqref{eq:gbf_theta} as $\theta_0$ which is determined by the original matrix \({\cal J}_J(\bm u_J^*,\xeq)\) and the perturbation matrix \(\Delta(\bm u_J^*,\bm u_K^*,\xeq)\). Let the eigenvalues of \(\tilde{\cal J}_\epsilon(\bm u_J^*,\bm u_K^*,\xeq)\) be denoted by \(\{\mu_i\}_{i=1}^N\) and those of \({\cal J}_J(\bm u_J^*,\xeq)\) by \(\{\lambda_i\}_{i=1}^N\). Because the infinitesimal parameter corresponding to $\delta$ in Eq.~\eqref{eq:gbf_theta} is $\epsilon$, from the generalized Bauer--Fike theorem, for any $\mu_i$ there exists $\lambda_j$ such that
\begin{equation}
    |\mu_i-\lambda_j| \leq \max\{\epsilon\theta_0,(\epsilon\theta_0)^{1/q}\}\label{eq:bf}
\end{equation}
holds. Therefore, we can take $\epsilon$ small enough so that if we have \(\Re\lambda_j<0\) then \(\Re \mu_i<0\) holds.


Although Eq.~\eqref{eq:bf} also holds for zero eigenvalues, tighter bounds on the zero eigenvalues are necessary for the stability analysis because there is a possibility that some of the zero eigenvalues in \({\cal J}_J(\bm u_J^*,\xp)\) become non-zero by perturbations.\footnote{Since $\xeq$ is globally asymptotically stable in SCC, ${\cal J}_{J}(\bm u^*_J,\xeq)$ has no pure imaginary eigenvalue.} We show that the zero eigenvalues of \({\cal J}_J(\bm u_J^*,\xeq)\) remain zero after infinitesimally small perturbations, that is, the zero eigenvalues of ${\cal J}_J(\bm u_J^*,\xeq)$ remain zero in ${\cal J}_J(\bm u_J^*,\xp)+{\cal J}_K(\bm u_K^*,\xp)$. 

Recall that we have $\Im \mathbb S=\Im \mathbb S_J$ (Eq.~\eqref{eq:same_SCC}). We rewrite $\Im\Bigl[{\cal J}_J(\bm u_J^*,\xp)+{\cal J}_K(\bm u_K^*,\xp)\Bigr]$ as follows
\begin{eqnarray}
    \Im\Bigl[{\cal J}_J(\bm u_J^*,\xp)+{\cal J}_K(\bm u_K^*,\xp)\Bigr]&=&\Im\Bigl[\mathbb S_J\,\mathrm{diag}\{\bm u_J^*\}\frac{\partial \bm p_J}{\partial \x}(\xp) + \mathbb S_K\,\mathrm{diag}\{\bm u_K^*\}\frac{\partial \bm p_K}{\partial \x}(\xp)\Bigr]\nonumber\\
    &=&\Im\Bigl[\begin{pmatrix}\mathbb S_J&\mathbb S_K\end{pmatrix}\mathrm{diag}\{\bm u_J^*,\bm u_K^*\}\begin{pmatrix}\frac{\partial \bm p_J}{\partial \x}\\\frac{\partial \bm p_K}{\partial \x}\end{pmatrix}(\xp)\Bigr]\\
    &=&\Im\Bigl[\mathbb S\mathrm{diag}\{\bm u^*\}\frac{\partial \bm p}{\partial \x}(\xp)\Bigr]\\
    &\eqcolon&{\Im} {\cal J}(\bm u^*,\xp)
\end{eqnarray}
where $(\mathbb S_J\,\mathbb S_K)$ and $({\partial \bm p_J}/{\partial \x}\,{\partial \bm p_K}/{\partial \x})^\top$ are the horizontal and vertical concatenation of the matrices, respectively. $\mathrm{diag}\{\bm u_J^*,\bm u_K^*\}$ is the diagonal matrix of the concatenated vector of $\bm u_J^*$ and $\bm u_K^*$. $\mathbb S$, $\bm u$, and $\bm p$ are the stoichiometric matrix, control input vector, and thermodynamic part of the reaction function vector consisting of all reactions, respectively. As ${\cal J}_J(\bm u_J^*,\xp)+{\cal J}_K(\bm u_K^*,\xp)$ and $ {\cal J}(\bm u^*,\xp)$ are equivalent as a linear map, we deal with ${\cal J}(\bm u^*,\xp)$ in this section.

In the following, we denote $\mathrm{diag}\{\bm u_X^*\}\frac{\partial \bm p_X}{\partial \x}(\x)$ as $A_X(\x)$. $A(\x)$ without subscript represents $\mathrm{diag}\{\bm u^*\}\frac{\partial \bm p}{\partial \x}(\x)$. The goal of this section is to show the existence of the state $\xp\in C(\bm \sigma)\cap{\cal W}(\xeq)$ being sufficiently close to $\xeq$ such that the following equation holds 
\begin{equation}
    {\Im}{\cal J}(\bm u^*,\xp)={\Im} {\cal J}_J(\bm u_J^*,\xeq).
\end{equation} If this holds, we have $\ker {\cal J}(\bm u^*,\xp)=\ker {\cal J}_J(\bm u_J^*,\xeq)$, and thus, for any eigenvectors of ${\cal J}_J(\bm u_J^*,\xeq)$ corresponding to the zero eigenvalue, we have ${\cal J}(\bm u^*,\xp)\bm v=\bm 0$. Therefore, the eigenvalues of ${\cal J}(\bm u^*,\xp)$ that correspond to these vectors remain zero. 

To this end, we show 
\begin{equation}
    {\Im}{\cal J}(\bm u^*,\xeq)={\Im} {\cal J}_J(\bm u_J^*,\xeq)\label{eq:im_consition_toshow}
\end{equation} holds (note that the input of ${\cal J}$ here is $\xeq$, not $\xp$). Since $\xeq$ is asymptotically stable within the SCC, $\rank \mathbb S_J=\rank {\cal J}_J(\bm u_J^*,\xeq)$ holds; otherwise, $\xeq$ cannot attract states from the entire SCC. Thus, the matrix rank of $A_J(\xeq)$ and $\mathbb S_J$ satisfy $\rank A_J(\xeq)\geq\rank \mathbb S_J$. Further, $$\rank A_J(\xeq)= \rank \mathbb S_J$$ holds because $A_J$ is $|J|\times N$ matrix, $|J|$ is less than or equal to $N$, and $\rank \mathbb S_J$ is $|J|$. 

The matrix $A_J(\xeq)$ is surjective as a linear map $\mathbb R^N\to \mathbb R^M$, and thus, we have 
\begin{equation}
{\Im} \mathbb S_J A_J(\xeq)={\Im} \mathbb S_J. \label{eq:rank1}
\end{equation}

As $\Im\mathbb S=\Im\mathbb S_J$, the condition to be proven (Eq.~\eqref{eq:im_consition_toshow}) is equivalent to the following conditions
\begin{eqnarray}
    &&{\Im}{\cal J}(\bm u^*,\xeq)={\Im} {\cal J}_J(\bm u_J^*,\xeq)\\
    &\Leftrightarrow&{\Im}\mathbb SA(\xeq)={\Im}\mathbb S_JA_J(\xeq)\\
    &\Leftrightarrow&{\Im}\mathbb SA(\xeq)={\Im}\mathbb S_J\\
    &\Leftrightarrow&{\Im}\mathbb SA(\xeq)={\Im}\mathbb S.
\end{eqnarray}
Thus, it is sufficient to show 
\begin{equation}{\Im}\mathbb SA(\xeq)={\Im}\mathbb S\label{eq:rank_condition1}.\end{equation} 
As Eq.~\eqref{eq:rank_condition1} is equivalent with 
\begin{equation} \rank A(\xeq)\geq \rank \mathbb S,\label{eq:rank_condition2}\end{equation}
we show Eq.~\eqref{eq:rank_condition2} as follows. Note that the $(n,r)$ element of the $A_X(\x)$ matrix is given by $u_r\partial p_r/\partial x_n(\x)$ where $r$ is the $r$th index of $X$ which is either nothing or $J$. Since $J$ is a subset of $I$, all the column in $A_J(\x)$ is in $A(\x)$. So, $\Im A(\x)\supseteq \Im A_J(\x)$ holds.

By using $\Im A(\x)\supseteq \Im A_J(\x)$ and $\rank A_J(\xeq)=|J|$, we have
\begin{equation}
\rank A(\xeq)=\dim\Im A(\xeq)\geq \dim \Im A_J(\xeq)=\rank A_J(\xeq)=|J|=\rank \mathbb S.\nonumber
\end{equation}
Therefore, we have $\rank A(\xeq)\geq \rank \mathbb S$, and so, ${\Im} {\cal J}_I(\bm u_I^*,\xeq)={\Im}{\cal J}_J(\bm u_J^*,\xeq)$. 

Since $\mathbb S$ and $\mathbb S_J$ are constant matrices and $A(\x)$ and $A_J(\x)$ are continuous matrices of $\x$, we have $\xp$ in $N_\epsilon(\xeq)\cap C(\bm \sigma)$ such that ${\Im} {\cal J}_I(\bm u_I^*,\xp)={\Im}{\cal J}_J(\bm u_J^*,\xeq)$ holds, where $N_\epsilon(\xeq)$ is the $\epsilon$ neighborhood of $\xeq$ restricted to SCC, that is,
$$N_\epsilon(\xeq):=\{\x\in{\cal W}(\xeq)\mid \|\bm x-\xeq\|<\epsilon\}.$$

\subsection{Summary of the proof}
We have shown the existence and uniqueness of the equilibrium state $\xeq$ on the boundary of the free cell $C(\bm \sigma)$ and the local asymptotic stability of $\xp\in C(\bm \sigma)$ in Eq.~\eqref{eq:eq2} with the control $\bm u^*$. Thus, first, we can control the system from $\bm x^{\rm src}\in {\cal W}(\bm x^{\rm tgt})$ to $\xeq$ with the constant control satisfying 
\begin{equation}
    \bm u=\begin{cases}>0&(i\in J)\\
    =0&(i\in K)\end{cases}.
\end{equation}
Next, we switch the control to $\bm u_J=\bm u_J^*$ and $\bm u_K=\bm u_K^*$. The state converges to a state being arbitrarily close to $\xp$, in particular, inside the free cell. Once the state enters the free cell, we can take advantage of the unrestricted controllability inside the free cell to control the state to $\bm x^{\rm tgt}$. The proof of theorem~\ref{thm:main} is completed.

It is noteworthy that the asymptotic stability of the equilibrium state is utilized twice; first, to show the existence of the control to $\xeq$ and second, to show the stability of the equilibrium state for checking the stability of \(\tilde{\cal J}_\epsilon(\bm u_J^*,\bm u_K^*,\xeq)\). Also note that the coincidence of the controllable set and the SCC holds as long as the global asymptotic stability of the detailed-balancing, adherent point of the target free cell $C(\bm \sigma)$ holds. Therefore, our argument never exclude the possibility that models with non-SCK also have the universal controllability. 

In the proof, we supposed that all reactions are reversible, i.e, $\bm k\in \mathbb R^R_{>0}$. However, this is not necessary. If reactions in the index set $H$ are irreversible, the balance manifolds of reactions in $H$ do not exist. However, the arguments above hold by replacing the index set of the reactions $I$ by $I\backslash H$. If there is a full SCK subset that indices of corresponding reactions are chosen from $I \backslash H$, we have an equilibrium state, $\xeq$, as an adherent point of the free cell $C(\bm \sigma)$, and it is globally asymptotically stable. The following part of the proof does not depend on the reversibility of the reactions.

\subsection{Cases with $\bm \sigma\notin \{-1,1\}^R$}
Here, we consider cases where the target state $\bm x^{\rm tgt}$ is in the free cell $C(\bm \sigma)$ with $\sigma_i=0$ for some indices. Note that there is the upper limit of the count of the zero elements in $\bm \sigma$. The $r$th element of $\bm \sigma$ being zero means that the cell $C(\bm \sigma)$ is a subset of the balance manifold of the $r$th reaction. The maximum number of intersecting balance manifold is given by $\rank \mathbb S\leq N$. 

For a clear statement of the condition, we introduce the \textit{parental cell}
\begin{defi}[Parental cell]
A cell $C(\bm \eta)$ is the \textit{parental cell} of the cell $C(\bm \sigma)$ if the sign vector $\bm \eta$ is given by
    \begin{equation}
        \eta_i=\begin{cases}
            1\ {\rm or}\ -1& (\sigma_i=0)\\
            \sigma_i& (\sigma_i\neq 0).
        \end{cases}
    \end{equation}
\end{defi}
The parental cell of $C(\bm \sigma)$ is generally not unique while the parental cell of $C(\bm \sigma)$ with $\bm \sigma\in \{-1,1\}^R$ is $C(\bm \sigma)$ itself. 

Let $L_{\bm \sigma}$ be the index set of zero elements in $\bm \sigma$, that is, $L_{\bm \sigma}\coloneq\{i\in I\mid \sigma_i=0\}$. For proving the controllability of a free cell $C(\bm \sigma)$ with $\bm \sigma\in \{-1,0,1\}^R$, we assume that there is at least one parental cell of the free cell $C(\bm \sigma)$, $C(\bm \eta)$, with a full SCK subset $\{{\cal M}_j\}_{j\in J}$ where $L_{\bm \eta}\subset J$ holds. Also, we denote the set difference of $I$ and $J$ by $K$, i.e., $K\coloneqq I\backslash J$. 

Under the condition described above, we can show the existence and uniqueness of the equilibrium state $\xeq\coloneq{\cal W}(\bm x^{\rm tgt})\cap {\cal M}_J$ in the identical manner as in the previous proof. Therefore, the equilibrium state $\xeq$ exists uniquely and is globally asymptotically stable in the SCC.

Next, we define $\xp$ so that 
$$\xp=\xeq +\epsilon \bm c\,,\xp\in C(\bm \sigma)\subset{\cal M}_{L_{\bm \sigma}},\,(\|c\|=1,\epsilon\ll1)$$
holds. Note that such $\bm c$ exists for arbitrarily small $\epsilon$ because $\xeq\in {\cal M}_J\subset {\cal M}_{L_{\bm \sigma}}$ (Recall that $J$ is chosen so that $L_{\bm \sigma}\subset J$ holds.). By expanding the right-hand side of Eq.~\eqref{eq:eq1} on $\epsilon$ in the same manner as in the previous proof, 
\begin{eqnarray}
    \dv{\Delta \x}{t} &=& \Bigl[\mathbb S_{J\backslash L_{\bm \sigma}} \bm u_{J\backslash L_{\bm \sigma}}(t)\odot \bm p_{J\backslash L_{\bm \sigma}}(\xp) +\mathbb S_K \bm u_K(t)\odot \bm p_K(\xp)\Bigr] \nonumber \\
    &&+\Bigl[\mathbb S_{J}\,\mathrm{diag}\{\bm u_{J}(t)\}\frac{\partial \bm p_{J}}{\partial \x}(\xp) + \mathbb S_K\,\mathrm{diag}\{\bm u_K(t)\}\frac{\partial \bm p_K}{\partial \x}(\xp)\Bigr]\Delta \x\label{eq:eq2_sigmazero}.
\end{eqnarray}
Note that $p_l(\xp)=0$ holds for $l\in L_{\bm \sigma}$ because $\xp \in {\cal M}_{L_{\bm \sigma}}$ while the derivative $\partial p_{l}/\partial \bm x\,(l\in {L_{\bm \sigma}})$ is a non-zero vector. Since $C(\bm \sigma)$ is the free cell, the stoichiometric vectors $\{\sigma_i\bm S\}_{i\in I\backslash {L_{\bm \sigma}}}$ are positively dependent, and thus, there exists a control vector $\bm u$ such that the first term of Eq.~\eqref{eq:eq2_sigmazero} vanishes. Also, such control vector can be chosen so that $\bm u_{J\backslash L_{\bm \sigma}}$ and $\bm u_K$ be ${\cal O}(1)$ and ${\cal O}(\epsilon)$ for $\epsilon\to 0$, respectively. The second term of Eq.~\eqref{eq:eq2_sigmazero} is identical to that in Eq.~\eqref{eq:eq1}. Thus, we can apply the same argument as that in the previous proof. 

As $C(\bm \sigma)$ is a free cell, we can utilize the unrestricted controllability inside the free cell to control the state to $\bm x^{\rm tgt}$. 

The statement of the main theorem is summarized as follows. 

\begingroup
\setcounter{thm}{0}
\makeatletter
\renewcommand{\thethm}{S\arabic{thm}'}
\makeatother
\begin{thm}[Main theorem with ternary $\bm \sigma$]\label{thm:main_ter}
    For a given free cell $C(\bm \sigma)$, if there is a parental cell $C(\bm \eta)$ having a full SCK subset within its boundaries $\{{\cal M}_j\}_{j\in J}$ satisfying $L_{\bm \eta}\subset J$, then the controllable set of any state in the free cell $\bm x\in C(\bm \sigma)$ coincides with the entire stoichiometric compatibility class, $\frakC(\bm x)={\cal W}(\bm x)$.
\end{thm}
If all reactions are implemented by SCK, every cell has at least $\rank \mathbb S$ balance manifolds as its boundary and full SCK subsets, Cor.~\ref{col:main_whole} and \ref{cor:local_global} are unmodified.

\subsection{Proof of the Prop.~\ref{prop:facet}}

\begingroup
\setcounter{prop}{2}
\begin{prop}
Suppose \(C(\bm \sigma)\) is non-empty. Then one can construct an index set \(J \subset I\) with $|J|=\rank\mathbb T$ such that:
\begin{enumerate}
  \item The matrix \(\mathbb{T}_J\) satisfies $\rank \mathbb{T}_J \;=\;\rank \mathbb{T}.$
  \item The balance manifolds \(\{\mathcal{M}_j\}_{j\in J}\) are boundary of the cell $C(\bm \sigma)$.
\end{enumerate}
\end{prop}
\endgroup

\begin{proof}
Without loss of generality we assume \(\sigma_i=1\) for all \(i\).

As the closure of the cell is the polyhedron in the logarithm-transformed space, we rewrite the variables and parameters as \(\bm y = -\ln\bm x\) and \(\bm b = \ln\bm k\). Then, the closure of the cell is given by 
\[
  \overline{C(\bm \sigma)}
  = \bigl\{\bm y\in\mathbb{R}^N \mid \mathbb{T}^\top\,\bm y \;\le\;\bm b \bigr\}.
\]
Since \(C(\bm\sigma)\neq\emptyset\), this system has an interior point and is full-dimensional.\footnote{In general, it is possible that the combination of inequalities leads to the equality constraint (e.g. \(\bm a^\top\bm x\le b\) and \(-\bm a^\top\bm x\le -b\) together enforce \(\bm a^\top\bm x=b\).). Such resulting equalities are termed \textit{implicit equalities}. Full-dimensional polyhedra are the polyhedra without implicit equalities. Full-dimensional polyhedra have the same dimension as the space that the polyhedra are embedded.}

The $(N-1)$ dimensional boundaries of $N$ dimensional polyhedra are termed \textit{facet} and there is a useful theorem for identifying them: 
\begin{thm}[Theorem 8.1 in \cite{schrijver1998theory}]\label{thm:facets}
Let
\(\mathbb{A}\bm x\le\bm b\)
be a system of non-redundant inequalities in \(\mathbb{R}^N\).  Then there is a one-to-one correspondence between the facets of the polyhedron
\[
  P = \{\bm x\in\mathbb{R}^N \mid \mathbb{A}\bm x\le \bm b\}
\]
and the individual inequalities in \(\mathbb{A}\bm x\le\bm b\).  Explicitly, the \(i\)th facet is
\[
  F_i = \{\bm x\in P \mid \bm a_i^\top\bm x = b_i\},
\]
where \(\bm a_i\) is the \(i\)th row of \(\mathbb{A}\).
\end{thm}

Here the inequalities \(\mathbb{A}\bm x\le\bm b\) are \emph{redundant} if removing one inequality from \(\mathbb{A}\bm x\le\bm b\) does not change the polyhedron. In general, \(\mathbb{T}\bm y\le\bm b\) may contain redundancies, so we apply Farkas' Lemma to eliminate them:

\begin{prop}[Farkas' Lemma]
Let \(\mathbb{A}\) be an $\mathbb R^{m\times d}$ matrix and let \(\bm \alpha \in\mathbb{R}^d\), \(\beta\in\mathbb{R}\). Then, \(\bm \alpha^\top\bm x\le \beta\) holds for all \(\bm x\) satisfying \(\mathbb{A}\bm x\le\bm b\) if and only if there exists \(\bm\lambda\ge\bm0\) such that \(\bm\lambda^\top\mathbb{A}=\bm \alpha^\top\) and \(\bm\lambda^\top\bm b \le \beta\).
\end{prop}

By iteratively removing the redundant inequalities, one obtains the non-redundant inequalities defining the polyhedron
\begin{equation}
  \mathbb{T}^*\,\bm y\le \bm b^*. \label{eq:irredundant_ineq}
\end{equation}
From Farkas' lemma, the rank of the matrix $\mathbb T$ does not change by the removal of redundant inequalities. Thus, \(\rank \mathbb{T}^* = \rank \mathbb{T}\) holds. From Thm.~\ref{thm:facets}, there are at least \(\rank\mathbb{T}\) facets of \(\overline{C(\bm \sigma)}\), each corresponding to a balance manifold. One can construct the index set $J$ by selecting the reaction index such that the corresponding column vectors of \(\rank\mathbb{T}^*\) become all linearly independent. Then, $|J|=\rank \mathbb T_J=\rank \mathbb T$ and the balance manifolds $\{{\cal M}\}_{j\in J}$ correspond to the boundary of the cell $C(\bm \sigma)$.
\end{proof}

\section{Derivation of the non-SCK from SCK reactions}

Let us consider the reaction $$2X\leftrightharpoons Y,$$ catalyzed by enzyme E. 
An ordinary reaction diagram for this reaction is shown in Fig.~\ref{fig:diagram}(a), which leads to the following coarse-grained reaction kinetics when the quasi-equilibrium method is applied:
$$v\propto\frac{[X]^2-k[Y]}{K_M+[X]/K_{X,1}+[X]^2/K_{X,2}+[Y]/K_Y},$$
where $[\cdot]$ is the concentration of the corresponding chemical. $K_*$s are the Michaelis-Menten constant set by the forward and backward reaction rate constants of each elementary step. $k$ is the reversibility parameter. Since we focus only on the function form, we do not go into the details of the parameters. The point here is that the thermodynamic part of this kinetics is given by $[X]^2-k[Y]$, and the reaction order matches the stoichiometry of each chemical, i.e., the reaction rate function is stoichiometrically compatible kinetics (SCK).

Another reaction diagram is shown in Fig.~\ref{fig:diagram}(b), where the irreversible maturation process $EX\to E^*X$ is introduced. Owing to the irreversibility of the maturation process, the process $E+X\leftrightharpoons EX\to E^*X \leftrightharpoons E^*XX\to E+Y$ becomes irreversible \footnote{Even if the final step is made reversible, the following argument does not change qualitatively. We set the final step irreversible for the ease of calculation.}. Thus, we must introduce another diagram for the backward reaction $Y\to 2X$ to make the coarse-grained reaction reversible. The diagram for the backward reaction is the right side of the diagram. The dynamics of the concentrations is given by the following equation

\begin{eqnarray}
    \dv{[E]}{t}&=&-k^a_+[E][X]+k^a_-[EX]+v[E^*XX]-l^a_+[E][Y]+(l^a_-+w)[EY]\\
  \dv{[EX]}{t}&=&k^a_+[E][X]-k^a_-[EX]-k^b_+[EX]\\
  \dv{[E^*X]}{t}&=&k^b_+[EX]-k^c_+[E^*X][X]+k^c_-[E^*XX]\\
  \dv{[E^*XX]}{t}&=&k^c_+[E^*X][X]-k^c_-[E^*XX]-v[E^*XX]\\
  \dv{[EY]}{t}&=&l^a_+[E][Y]-l^a_-[EY]-w[EY].
\end{eqnarray}
The \textit{de-novo} production of the enzyme is not considered within the timescale that we focus on, and the total concentration of the enzyme is conserved: $$[E]+[EX]+[E^*X]+[E^*XX]+[EY]=[E]_T={\rm const.}$$ 

The steady-state solutions of the final complexes $E^*XX$ and $EY$ are given by 
\begin{eqnarray}
{[E^*XX]}&=&\frac{k^a_+k^b_+}{v(k^a_-+k^b_+)}[E][X]\\
{[EY]}&=&\frac{l^a_+}{l^a_-+w}[E][Y]
\end{eqnarray}
Thus, the net reaction flux is given by 
$$\tilde{v}=v[E^*XX]-w[EY]=\tilde{v}_{\rm max}[E]\bigl([X]-\tilde{k}[Y]\bigr),$$
where $\tilde{v}_{\rm max}=\frac{k^a_+k^b_+}{k^a_-+k^b_+}$ and $\tilde{k}=\frac{wl^a_+}{l^a_-+w}\frac{k^a_-+k^b_+}{k^a_+k^b_+}$. Note that $[E]$ is the concentration of the free enzyme and $\tilde{v}_{\rm max}[E]$ leads to the kinetic part of the reaction rate function. The thermodynamic part is given by $[X]-\tilde{k}[Y]$, and thus, it is non-SCK. 

In this derivation, the enzyme must recognize the metabolite with which the complex is formed, and the subsequent reaction is determined. In addition, for non-SCK kinetics, the reactions $EX\to E^*X$ and $EY\to E+2X$ must be irreversible. Thus, this reaction $2X\leftrightharpoons Y$ needs to be externally driven.

\begin{figure}[htbp]
    \begin{center}
    \includegraphics[width = 160 mm, angle = 0]{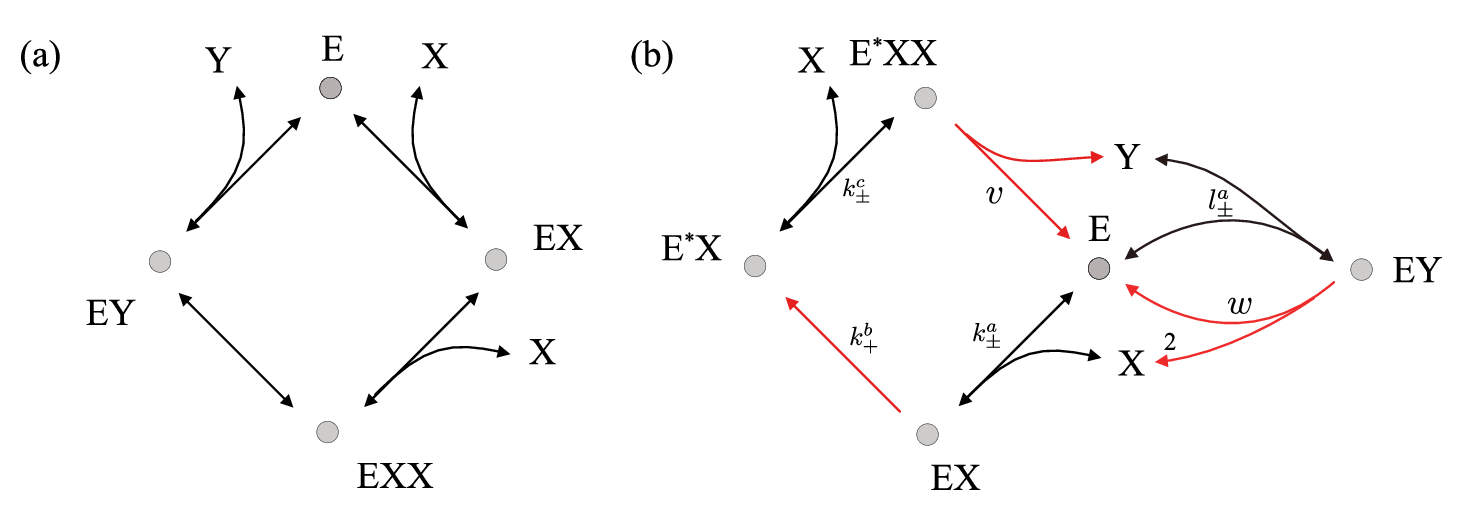}
    \caption{(a) A reaction diagram for SCK. (b) A reaction diagram for non-SCK. The coarse-grained reaction scheme is the same with (a), $2X\leftrightharpoons Y$. The black arrows and red arrows represent the reversible reactions and irreversible reactions, respectively. ``2'' on the red arrow on the right bottom represents that two molecules of X are produced by the reaction. The corresponding symbols of rate constants of the reactions are put on the arrows in (b).}
        \label{fig:diagram}
      \end{center}
    \end{figure}

{\color{black}
\section{Construction of the sub-network and its link to thermodynamic consistency}
We have discussed the relationship between Stoichiometrically Compatible Kinetics (SCK) and thermodynamic consistency in the main text. In this section, we clarify how the active sub-network used in our proof is constructed and, more importantly, why our results do not rely on global parameter constraints often associated with thermodynamic consistency. In this section, we term the model with some reactions turned-off for the relaxation to the equilibrium state the \textit{reduced model}.

Specifically, while thermodynamic consistency in a cyclic network generally requires the kinetic parameters to satisfy the detailed-balance condition at equilibrium, our theorem ensures global controllability even when such conditions are violated. We illustrate this point using the Onsager model as an example.

The Onsager model is given by the following reactions:
\begin{eqnarray*}
    \text{R}_1&:&\text{A}_{\rm ext}\leftrightharpoons\text{A}\\
\text{R}_2&:&\text{A}\leftrightharpoons \text{B}\\
\text{R}_3&:&\text{B}\leftrightharpoons \text{C}\\
\text{R}_4&:&\text{C}\leftrightharpoons \text{A}\\
    \text{R}_5&:&\text{C}_{\rm ext}\leftrightharpoons\text{C}
\end{eqnarray*}
the schematic of the reaction network is given by Fig.~\ref{fig:onsager}(a). The ODE with control is given by 
\begin{eqnarray}
\dv{}{t}\begin{pmatrix}
    a\\b\\c
\end{pmatrix}
=
\begin{pmatrix}
    1&-1&0&0&0\\
    0&1&-1&0&0\\
    0&0&1&-1&1
\end{pmatrix}
\begin{pmatrix}
u_1\\u_2\\u_3\\u_4\\u_5
\end{pmatrix}
\odot
\begin{pmatrix}
v_1\\v_2\\v_3\\v_4\\v_5
\end{pmatrix}
\end{eqnarray}
Here, we consider a simple case that the reaction rate function follows the mass-action kinetics:
    
\begin{eqnarray}
v_1&=&k_1^+a_{\rm ext}-k_1^-a,\\
v_2&=&k_2^+a-k_2^-b,\\
v_3&=&k_3^+b-k_3^-c,\\
v_4&=&k_4^+c-k_4^-a,\\
v_5&=&k_5^+c_{\rm ext}-k_5^-c.
\end{eqnarray}

The rank of the stoichiometric matrix is three, and thus, we need to turn off two reactions to obtain a reduced model that relaxes to the equilibrium state. The choice of the reactions turned off depends on the target of the control.

Note that the model has a reaction cycle $\text{A}\to\text{B}\to\text{C}\to\text{A}$ and its reverse. Let us consider the model with these three reactions, $\text{R}_2,\text{R}_3$ and $\text{R}_4$. The detailed-balance condition for the reaction cycle is given by 
\begin{equation}
\frac{k_2^+k_3^+k_4^+}{k_2^-k_3^-k_4^-}=1. \label{eq:onsager_param}
\end{equation}
If this condition is not satisfied, the relaxation to equilibrium is prohibited, but the cycle current persists, i.e., the system relaxes to the non-equilibrium steady state. 

The remark of this section is that our theorem (Theorem.~\ref{thm:main} and \ref{thm:main_ter}) holds even if a model does not satisfy the condition on the parameters like Eq.\eqref{eq:onsager_param}. It is because for the construction of the reduced model, at least a single reaction in each cycle is turned off. According to Proposition.\ref{prop:facet}, for any non-empty cell $C(\bm \sigma)$ there is an index set $J$ such that the balance manifolds $\{{\cal M}_j\}_{j\in J}$ are the boundaries of the cell and the stoichiometric matrix consisting of the reactions in $J$, $\mathbb S_J$, satisfies ${\rm rank}\ \mathbb S={\rm rank}\ \mathbb S_J=|J|$. Therefore, the stoichiometric vectors of the reactions in $J$ must be linearly independent, and no cycle can be in the reaction network $\mathbb S_J$.

In the Onsager model, the rank of the full stoichiometric matrix is three, and thus, we need to select three reactions to keep activated and the others are turned off. If we construct a reduced model with reaction $\text{R}_2,\text{R}_3$ and $\text{R}_4$, the stoichiometric matrix $\mathbb S_J$ is not full-ranked, but ${\rm rank }\ \mathbb S_J=2$. All allowed choices of the reactions so that ${\rm rank }\ \mathbb S={\rm rank }\ \mathbb S_J=|J|$ are illustrated in Fig.~\ref{fig:onsager}(b). Readers can see that every reduced model has the detailed-balanced equilibrium as its steady-state. 

Two choices $\{\text{R}_2,\text{R}_3,\text{R}_4\}$ and $\{\text{R}_1,\text{R}_4,\text{R}_5\}$ do not satisfy the rank condition. Consistently, the steady-state of the reduced model with either $\{\text{R}_2,\text{R}_3,\text{R}_4\}$ or $\{\text{R}_1,\text{R}_4,\text{R}_5\}$ is, in general, the non-equilibrium steady-state (Fig.\ref{fig:onsager}(c)).

\begin{figure}[htbp]
    \begin{center}
    \includegraphics[width = 150 mm, angle = 0]{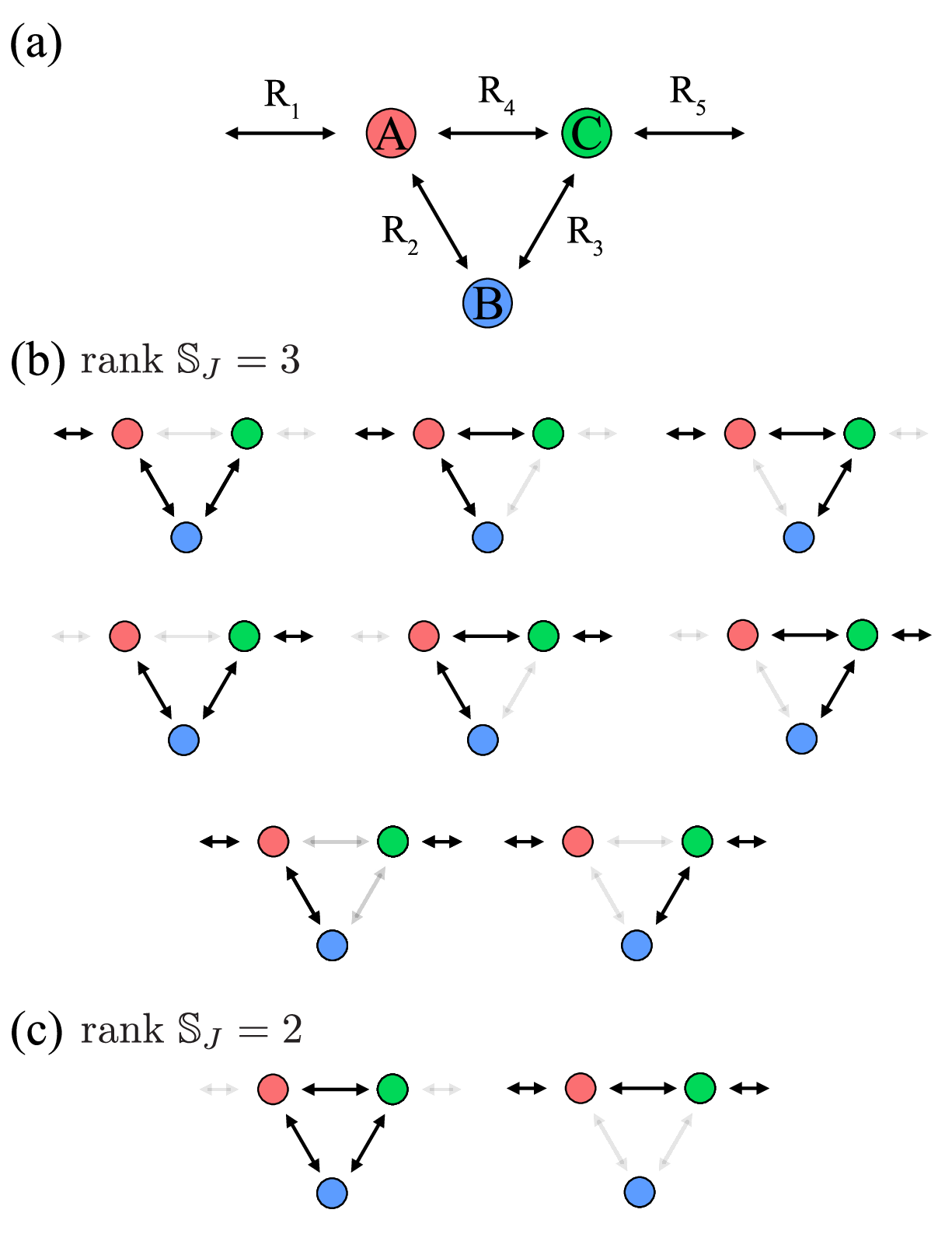}
    \caption{\tcr{(a) A reaction network of the Onsager model. (b) The allowed combinations of reactions to be kept active for constructing the reduced model that relaxes to the equilibrium state (${\rm rank}\mathbb S_J=3$), (c) and does not relax to the equilibrium state but to the non-equilibrium steady state (${\rm rank}\mathbb S_J=2$). The external A and C are not depicted for the ease of illustration.}}
        \label{fig:onsager}
      \end{center}
    \end{figure}

\tcr{While the detailed balance condition is not required for global controllability to states within free cells, it is nevertheless noteworthy that the detailed balance condition determines the number of free cells. 
Figure~\ref{fig:onsager_freecell} illustrates the free cells of the Onsager model for different choices of the reversibility parameters. In Fig.~\ref{fig:onsager_freecell}(a), the detailed balance condition is violated, and there exists a free cell in which a net reaction flow circulates along the cycle $A \to B \to C \to A$. 
By contrast, when the detailed balance condition is satisfied, only a single free cell remains, in which no net flow exists along the cycle as shown in Fig.~\ref{fig:onsager_freecell}(b).}

\begin{figure}[htbp]
    \begin{center}
    \includegraphics[width = 150 mm, angle = 0]{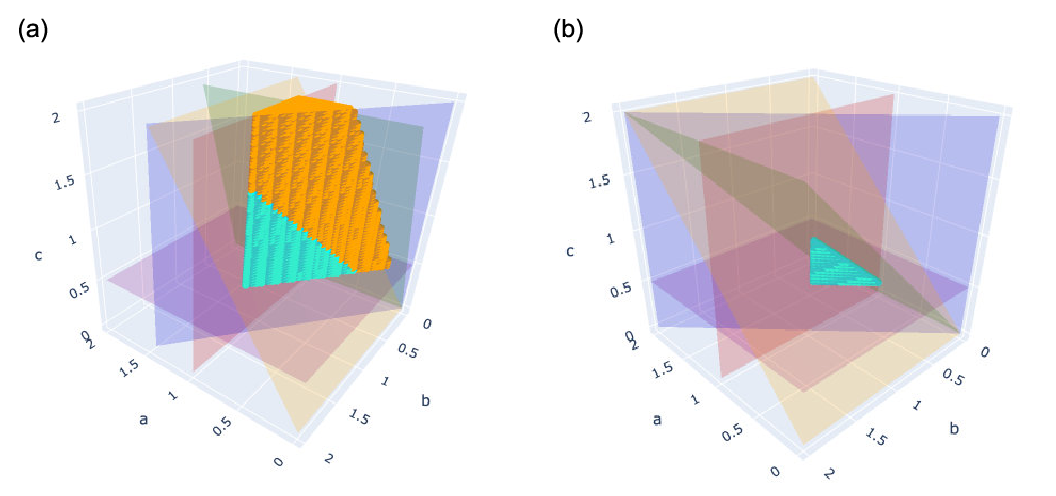}
    \caption{\tcr{The union of free cells in the Onsager model for different choices of the reversibility constants $\bm{k}$. (a) A case in which detailed balance is violated, 
$\left(k_2^+ k_3^+ k_4^+ / k_2^- k_3^- k_4^-\right) \neq 1$. The orange and cyan regions correspond to free cells. In the orange cell, the reaction direction is given by $(1,1,1,1,-1)$, indicating a net reaction flow along the cycle $A \to B \to C \to A$. In contrast, in the cyan cell, the reaction direction is $(1,1,1,-1,-1)$, for which no net flow exists along the cycle. (b) A case in which detailed balance is satisfied, $\left(k_2^+ k_3^+ k_4^+ / k_2^- k_3^- k_4^-\right) = 1$. In this case, there is only a single free cell, whose reaction direction is given by $(1,1,1,-1,-1)$. The red, blue, green, yellow, and purple planes are the balance manifolds of the reaction $\text{R}_1$, $\text{R}_2$, $\text{R}_3$, $\text{R}_4$, and $\text{R}_5$, respectively. The following parameters are identical in panels (a) and (b): $a_{\mathrm{ext}} = 1$, $c_{\mathrm{ext}} = 0.1$, $k_1^+ = 1.0$, $k_1^- = 1.0$, $k_5^+ = 5.0$, and $k_5^- = 1.0$. In panel (a), the parameters are set to $k_2^+ = 2.0$, $k_2^- = 1.5$, $k_3^+ = 3.0$, $k_3^- = 1.0$, $k_4^+ = 0.7$, and $k_4^- = 1.0$, whereas in panel (b) all of $k_2^+$, $k_2^-$, $k_3^+$, $k_3^-$, $k_4^+$, and $k_4^-$ are set to unity.
}}
        \label{fig:onsager_freecell}
      \end{center}
    \end{figure}

In this section we used the Onsager model as a simple model for highlighting the link between the control and thermodynamic consistency. However, note that in linear models such as the Onsager model, the controllability to the free cell is guaranteed from the chemical reaction network theory \cite{Feinberg2019}. We consider the model given by 
$$\dv{\bm x}{t}=\mathbb S\bm u(t)\odot \bm v(\bm x),$$
where $v(\bm x)$ is linear on $\bm x$ (monomolecular), and all reactions are reversible. According to the deficiency zero theorem, every reversible monomolecular reaction system has a unique fixed point in the positive stoichiometrically compatible class which is globally asymptotically stable for a given parameter set. Given the uniqueness of the fixed point and Prop.~\ref{prop:fixed_point}, the model has only a single free cell for a given parameters other than $\bm u$. Thus, for an arbitrarily selected constant control $\bm u(t)=\bm u^c \succ\bm 0$, all states converge to the attractor inside the free cell, and there is no other free cell. Since the control between any two states in the same free cell is possible, the controllable set of a given state in the free cell is a whole stoichiometric compatibility class.

\section{Absence of uncontrollable states in three biochemical models}
We have shown that the non-SCK kinetics is a key factor that the model has uncontrollable states to a free cell, and indeed, there were regions from which no control is possible to one of the two free cells in the toy model introduced in the section V in the main text. Here we show, however, implementing reactions with non-SCK kinetics does not necessarily result in the existence of uncontrollable states to a free cell. In this section we show the absence of such states in the popular biochemical models: Sel'kov model, Brusselator model, and Schnakenberg model.

The original Sel'kov model is a minimal model of the glycolitic oscillation, and is given by the following ODE
\begin{eqnarray}
\dv{x}{t}&=&b-(a+y^\gamma)x,\\
\dv{y}{t}&=&(a+y^\gamma)x -y.
\end{eqnarray}
The Sel'kov model is a model focusing on the positive feedback regulation of the phosphofructokinase (PFK) whose reaction stoichiometry is $$\text{fructose-6-phosphate}+\text{ATP}\leftrightharpoons\text{fructose-1,6-bisphosphate}+\text{ADP}.$$
In the Sel'kov model, the substrates and products of the reaction are lumped up to a single variable $X$ and $Y$, respectively. 

There are three reactions
\begin{eqnarray*}
    \text{R}_1&:&\emptyset\to \text{X}\\
\text{R}_2&:&\text{X}\to \text{Y}\\
\text{R}_3&:&\text{Y}\to \emptyset,
\end{eqnarray*}
where $\emptyset$ represents the external environment. The second reaction is the PFK reaction and PFK enzyme activity is positively regulated by Y (ADP) as $(a+y^\gamma)$. $a$ is the basal PFK activity and the activity increases as $y$ increase. For $\gamma>1$, there is a parameter region of $a$ and $b$ in which the Hopf bifurcation occurs and limit cycle attractor emerges.  

To study the controllability of the Sel'kov model, first we need to make the reactions reversible. In the original Sel'kov model, all reactions are irreversible, meaning that there is no balance manifold in $\mathbb R_{>0}^2$ and there is a single cell. Indeed, this cell is a free cell, and thus, control between arbitrarily chosen pair of states is always possible. The freeness of the cell is guaranteed by Proposition~\ref{prop:fixed_point}. The proposition states that a fixed point can exist only in free cells and if there is a fixed point in a cell, the cell is a free cell. The original Sel'kov model has the fixed point 
$$x=\frac{b}{a+b^\gamma},y=b$$
regardless of the parameter choice. Thus, the whole positive orthant  $\mathbb R_{>0}^2$ is a free cell.

To avoid this trivial controllability consequence, we need to introduce the reversibility to the model. We modify the reactions as follows;
\begin{eqnarray}
    \text{R}_1&:&\text{X}_{\rm ext}\leftrightharpoons \text{X}\label{eq:rxn_eq_selkov1}\\
\text{R}_2&:&\text{X}\leftrightharpoons \text{Y} \label{eq:rxn_eq_selkov2}\\
\text{R}_3&:&\text{Y}\leftrightharpoons \text{Y}_{\rm ext}. \label{eq:rxn_eq_selkov3}
\end{eqnarray}
We updated the model equation accordingly. By additing the control parameter $\bm u$, now the model equation is  given by 
\begin{eqnarray}
\dv{x}{t}&=&u_1(t)(b-x)-u_2(t)(a+y^\gamma)(x-k_2 y),\\
\dv{y}{t}&=&u_2(t)(a+y^\gamma)(x-k_2y) + u_3(t)(c-y),
\end{eqnarray}
where we have newly introduded parameters $c$ and $k_2$ representing the external concentration of Y, and the reversibility parameter of the reaction $R_2$, respectively. The reversibility parameter should also be introduced for the reaction $R_1$ and $R_3$, though the reversibility parameters of those reactions are considered to be absorbed to the external concentration and maximum reaction rate. The maximum reaction rate of each reaction is finally set to unity by absorbing it to the control parameter $u_i(t)$. 

Now the modified model is fully implemented by SCK. Here, we allow the reactions to be implemented by non-SCK 
\begin{eqnarray}
\dv{x}{t}&=&u_1(t)(b-x)-u_2(t)(a+y^\gamma)(x^\alpha-k_2 y^\beta),\label{eq:selkov_x}\\
\dv{y}{t}&=&u_2(t)(a+y^\gamma)(x^\alpha-k_2y^\beta) + u_3(t)(c-y).\label{eq:selkov_y}
\end{eqnarray}
The introduction of $\alpha$ and $\beta$ ($\alpha,\beta\geq 0$) allows us to make the reaction rate kinetics to deviate from SCK. Here, we did not make the reaction rate function for $R_1$ and $R_3$ be non-linear on $x$ and $y$, respectively. This is because even if we introduce nonlinearity, the balance manifold is inherently unchanged, that is, the balance manifold of the reaction $R_1$ with its rate function as $b-x$ is given by $${\cal M}_1=\{(x,y)\in\mathbb R^2_{>0}\mid x=b\},$$ while with rate function as $b-x^\delta$ is given by $${\cal M}'_1=\{(x,y)\in\mathbb R^2_{>0}\mid x=b^{1/\delta}\}.$$ Thus, the difference of ${\cal M}_1$ and ${\cal M}'_1$ is recovered by changing $b$ value, and introduction of $\delta$ does not qualitatively change the structure of the balance manifold.

In this extension we allowed the reaction order to be an arbitrary positive real number. However, we do not allow the chemical species which are not in the reaction equation to have non-zero value; for example, in the reaction $R_2$ (Eq.\eqref{eq:rxn_eq_selkov2}), we do not allow the reaction rate function to have a form like $(x^\alpha y^\mu-k_2 y^\beta)$ with $\mu\neq 0$ because Y is not a reactant of the forward reaction of $R_2$. Note that the substrate-level regulations are possible within this framework; for example, the reaction rate function of $R_2$ can have a form like $x^\lambda(a+y^\gamma)(x^\alpha-k_2 y^\beta)$, where the regulation term $x^\lambda(a+y^\gamma)$ modulates the enzyme activity but this part does not affect the controllability because it is the kinetic part of the reaction rate function and is absorbed to the control parameter $u_2(t)$.

The balance manifolds of Eq.\eqref{eq:selkov_x} and \eqref{eq:selkov_y} are given by 
\begin{eqnarray}
    {\cal M}_1&=&\{(x,y)\in\mathbb R^2_{>0}\mid x=b\} \label{eq:selkov_balance1}\\
    {\cal M}_2&=&\{(x,y)\in\mathbb R^2_{>0}\mid y=(x^\alpha/k_2)^{1/\beta}\}\\
    {\cal M}_3&=&\{(x,y)\in\mathbb R^2_{>0}\mid y=c\}\label{eq:selkov_balance3}
\end{eqnarray}
Note that the regulation term of the reaction $R_2$, $(a+y^\gamma)$, does not appear in the balance manifold ${\cal M_2}$. This is because the change in the overall rate (enzyme activity) is a part of the kinetic part in the decomposition of the reaction rate function into the thermodynamic and kinetic parts, $v(\bm x)=p(\bm x)f(\bm x)$. The kinetic part is absorbed into the control parameter $u$ and does not affect the controllability; only the thermodynamic part determines the structure of the cells and controllability.

In Eq.\eqref{eq:selkov_balance1}-\eqref{eq:selkov_balance3}, only ${\cal M}_2$ can change its functional shape while ${\cal M}_1$ and ${\cal M}_3$ are the parallel line to the Y-axis and X-axis, respectively. By chaning $\alpha$ and $\beta$, ${\cal M}_2$ can be convex, linear, and concave function, while it is the monotonically increasing function of $x$ regardless of the values of $\alpha,\beta$, and $k_2$ in $\alpha,\beta,k_2>0$. As shown in Fig.~\ref{fig:selkov}, the possible configuration of the cells is fully determined by the order that ${\cal M}_2$ has the intersection with ${\cal M}_1$ and ${\cal M}_3$ (e.g., first ${\cal M}_2$ intersects with ${\cal M}_1$ and then ${\cal M}_3$) regardless of whether ${\cal M}_2$ is convex, linear, or concave function because ${\cal M}_2$ is a monotonically increasing function of $x$.

\begin{figure}[htbp]
    \begin{center}
    \includegraphics[width = 150 mm, angle = 0]{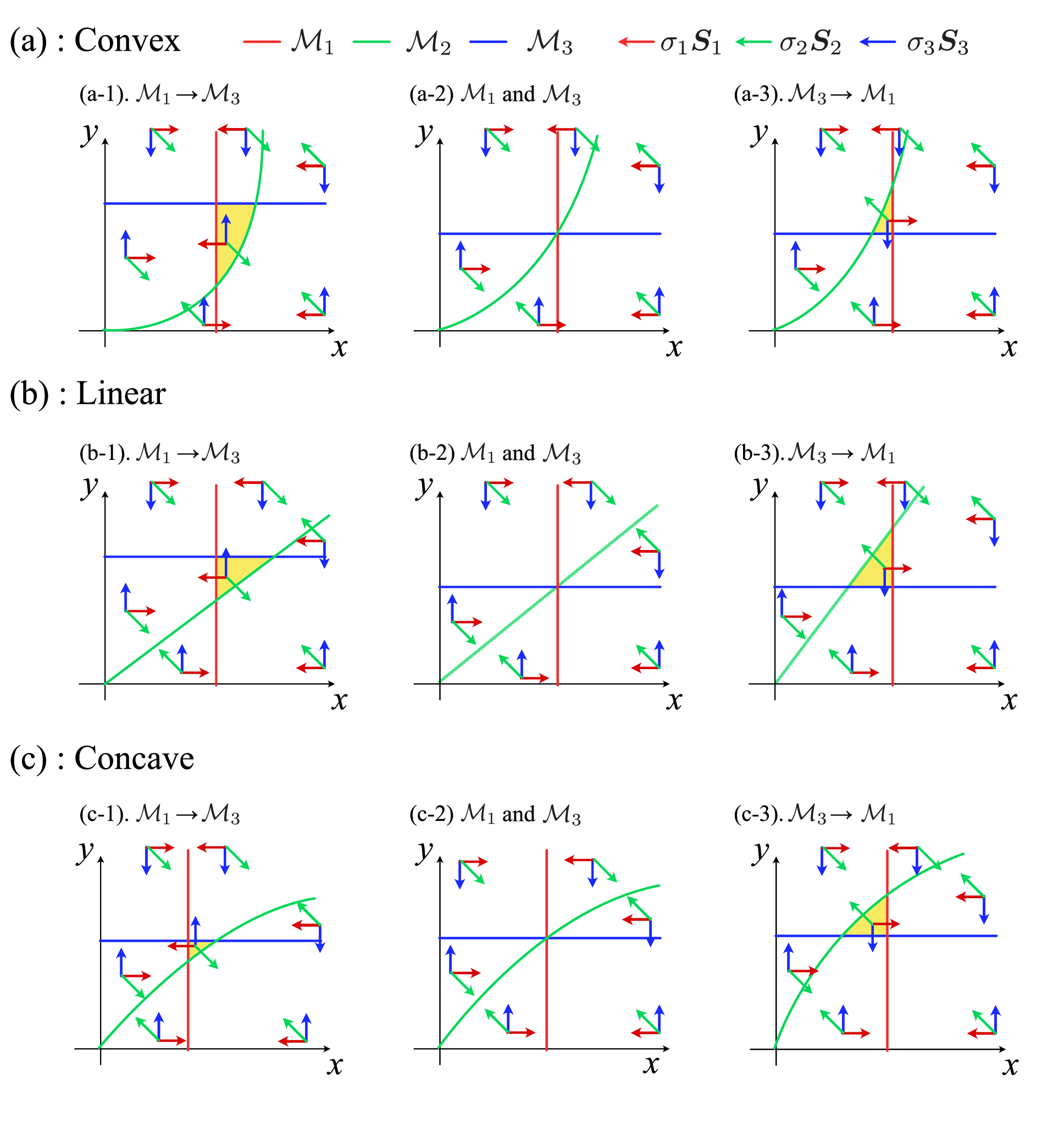}
    \caption{\tcr{All configurations of the cells by changing the parameters $a,b,c,\alpha,\beta$, and $k_2$ which are not control inputs but fixed constants. By changing $\alpha$ and $\beta$, the balance manifold ${\cal M}_2$ becomes the convex, linear, or concave function. The modulation of the parameters changes the order that ${\cal M}_2$ intersects with ${\cal M}_1$ and ${\cal M}_3$. When we follow ${\cal M}_2$ from the origin, the possible intersection patterns are as follows; (1) first intersects with ${\cal M}_1$ and then intersects with ${\cal M}_3$, (2) intersects with ${\cal M}_1$ and ${\cal M}_3$ at an indentical point, and (3) first intersects with ${\cal M}_3$ and then intersects with ${\cal M}_1$. In the figure, we showed the order of intersection for three different functional forms of ${\cal M}_2$, while the configuration of cells are the same once the order of intersection is determined, regardless of the functional form. The red, green, and blue lines or curve are the balance manifold ${\cal M}_1$, ${\cal M}_2$, and ${\cal M}_3$, respectively. The directed stoichiometric vectors of the reaction $R_1$, $R_2$, and $R_3$ are represented by the red, green, and blue arrows in each cell. The cell filled with yellow is the free cell.}}
        \label{fig:selkov}
      \end{center}
    \end{figure}

As shown in Fig.~\ref{fig:selkov}, the cell at the center is the only free cell in all cases. The case where ${\cal M}_2$ has an intersection simultaneously with ${\cal M}_1$ and ${\cal M}_3$ is an exception. In this case, only the intersection point is the free cell\footnote{By definition, a cell with a single point is a free cell because ${\rm span}\{\bm 0\}={\rm cone}\{\bm 0\}$. Also, since such a cell consists of only a single point, any points from the cell are mutually controllable.}. 

In the both cases of order of intersection, there is always a control to the free cell at the center. First, reaching a balance manifold ${\cal M}_3$ using the reaction $R_3$. If it is on the boundary of the free cell, by using the conical combination of the reaction $R_1$ and $R_2$, the state is controlled to the free cell. If it is not on the boundary of the free cell, by using only the reaction $R_1$, the state is controlled to the intersection of ${\cal M}_3$ and either ${\cal M}_1$ or ${\cal M}_2$, i.e., on a vertex of the free cell. As shown in the proof of the main theorem, the state is controlled into the free cell from a vertex of the free cell. Thus, the arbitrary state is controllable to the free cell.

The same argument is applied for the Schnakenberg model. The model is extended to have the reversible reactions, non-SCK rate functions, and the control as follows;

\begin{eqnarray}
\dv{x}{t}&=&u_1(t)(a-x)-u_2(t)(b-x)+u_3(t)x^2(y^\beta-k_3 x^\alpha),\label{eq:schnakenberg_x}\\
\dv{y}{t}&=&u_3(t)x^2(y^\beta-k_3 x^\alpha) + u_4(t)(c-y).\label{eq:schnakenberg_y}
\end{eqnarray}
In the Schnakenberg model, there are two exchange reactions of chemical X. As the function form of the third reaction is nonlinear $x^2(y^\beta-k_3 x^\alpha)$, the balance manifold can be convex, linear, or concave depending on $\alpha$ and $\beta$ values. But with the same argument that we have presented for the Sel'kov model, the configuration of the cell does not depend on whether the ${\cal M}_3$ is convex, linear, or concave. The all configurations of the cells are shown in Fig.~\ref{fig:schnakenberg}. The congulations of the cells are symmetric for the exchange of position of the balance manifold ${\cal M}_1$ and ${\cal M}_2$. Thus, we depict only the cases where ${\cal M}_2$ is located at higher value of X-coordinate than that of ${\cal M}_1$. An inherent difference in the Schnakenberg model is that it has two connected free cells. But in any cases, every state is controllable to any state in the free cells by utilizing the reaction $R_4$ to reach a sufficiently close state to ${\cal M}_4$, and then reach the free cells along ${\cal M}_4$.

\begin{figure}[htbp]
    \begin{center}
    \includegraphics[width = 160 mm, angle = 0]{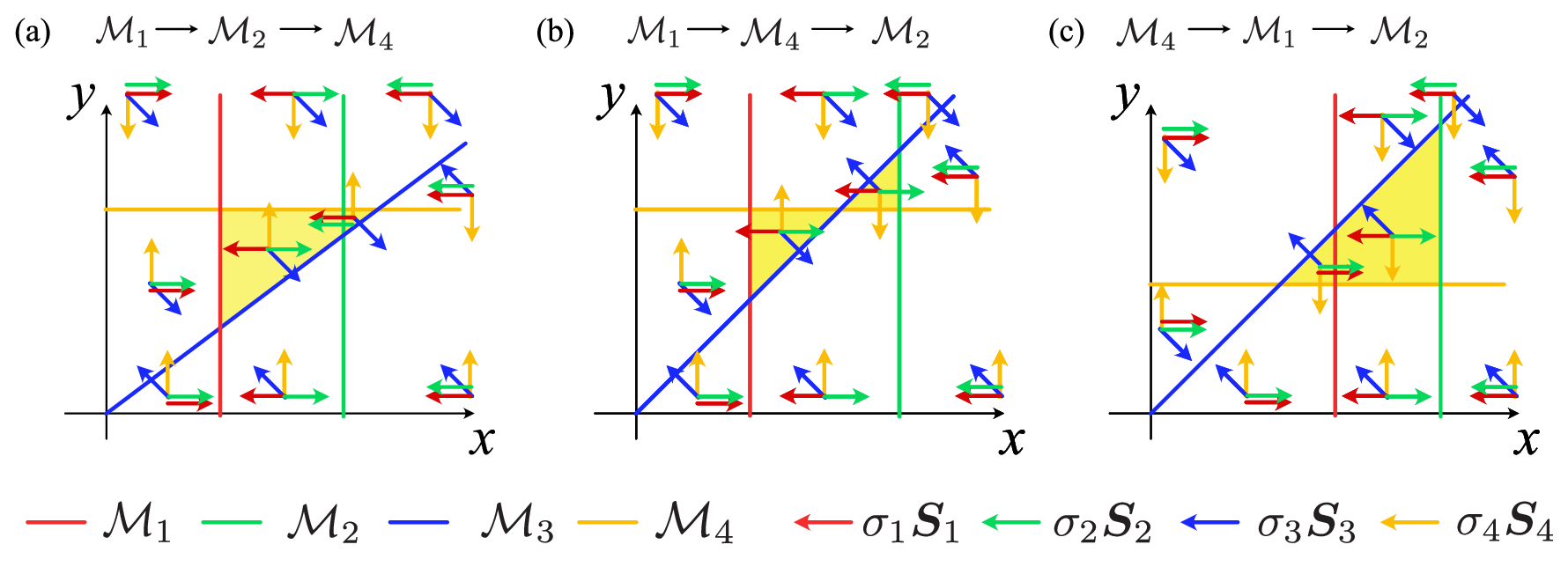}
    \caption{\tcr{The three configurations of the cells in the Schnakenberg model. The figure is made for the case that ${\cal M}_3$ is linear while the configuration does not change even if ${\cal M}_3$ is either convex or concave. Colored lines and arrows represent the corresponding balance manifolds and directed stoichiometric vectors, respectively. The regions filled with yellow are the free cells.}}
        \label{fig:schnakenberg}
      \end{center}
    \end{figure}

Finally, we check the Brusselator model. The reversible Brusselator model with control and non-SCK kinetics is given by 
\begin{eqnarray}
    \dv{x}{t} &=& u_1(t)(a-x)-u_2(t)(x^\alpha-k_1 y^\beta)+u_3(t)x^2(y^\gamma-k_3x^\delta)\\
    \dv{y}{t} &=& u_2(t)(x^\alpha-k_1 y^\beta)-u_3(t)x^2(y^\gamma-k_3x^\delta).
\end{eqnarray}
An important feature of the Brusselator is that the stoichiometric vector of the reaction $R_2$ and $R_3$ is identical except for the sign. $\bm S_2=-\bm S_3$. Because of this redundancy, there is no 2-dimensional free cell in the Brusselator model. All possible configurations of the cells are presented in Fig.~\ref{fig:brusselator}. The configurations are classified into three types based on where ${\cal M}_2$ and ${\cal M}_3$ have intersection in $\mathbb R_{>0}^2$: Fig.~\ref{fig:brusselator}(a) no intersection, (b) at the greater $x$ than $a$ (external concentration of X), and (c) at the lower $x$ than $a$. 

While there is no 2-dimensional free cell, the interval highlighted in yellow in the panels are ``free'' in terms that any two points are mutually controllable; in Fig.~\ref{fig:brusselator}(a) case for instance, for controlling a state to another state with lower $y$ value (note that both states are on the same interval), one can utilize the stoichiometric vector of the reaction $R_3$ to decrease $y$. Since this operation drags the state out from the interval, the stoichiometric vector of the reaction $R_1$ can recover the system's state onto the interval. Iteration of this operation allows us to control the system from $(a,y_1)$ to $(a,y_2)$ with $y_1>y_2$. Feasibility of the control to the opposite direction is confirmed by the same argument while we utilize the stoichimoetric vector of the reaction $R_2$ for this purpose. In this sense, the interval is a free cell, and indeed, the fixed points are only on this interval. 

By checking the configurations and cells and directed stoichiometric vectors, we can see that an arbitrary state is controllable to the interval. 


We have checked the controllability of the three popular 2-dimensional biochemical models. All of them have no state which is uncontrollable to a state in a free cell. In the low-dimensional models, models with uncontrollable state to states in free cells might be rare. Our toy model in the main text is dealt as a two-dimensional model, while the model consists of three chemical species and a conserved quantity allowed us to reduce the model into two-dimensional. As far as we have explored we have never found any models consisting of one or two chemical species and exhibit uncontrollability to a state in a free cell.

\begin{figure}[htbp]
    \begin{center}
    \includegraphics[width = 160 mm, angle = 0]{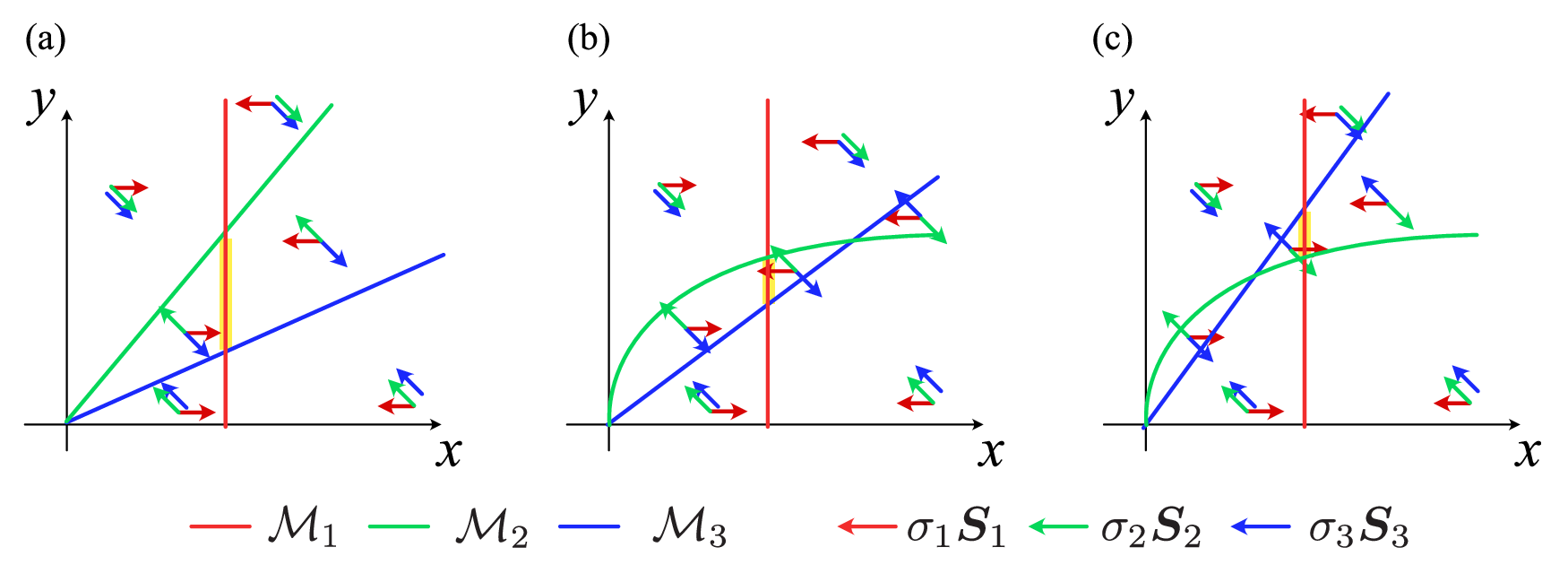}
    \caption{\tcr{The three configurations of the cells in the Brusselator model. The figure is made for the case that ${\cal M}_3$ is linear while the configuration does not change even if ${\cal M}_3$ is either convex or concave. Coloed lines and arrows represent the corresponding balance manifolds and directed stoichiometric vectors, respectively. The regions filled with yellow are the free cells.}}
        \label{fig:brusselator}
      \end{center}
    \end{figure}

    }

\end{document}